\documentclass[aos]{imsart}

\RequirePackage{amsthm,amsmath,amsfonts,amssymb}
\RequirePackage[numbers,sort&compress]{natbib}
\RequirePackage[colorlinks,citecolor=blue,urlcolor=blue]{hyperref}
\RequirePackage{graphicx}
\usepackage{hyperref}
\usepackage{cite}
\usepackage{natbib}
\usepackage{longtable}
\usepackage{comment}
\usepackage{tikz}
\usepackage{mathtools}
\usepackage{tabularx,ragged2e,booktabs,caption}
\usepackage{subcaption}
\usepackage{tikz}
\usepackage{pdflscape}
\usepackage{threeparttable}
\usepackage{prodint}
\usepackage{mathrsfs}
\usepackage{rotating}
\usepackage{booktabs}
\usepackage{pgfplots}
\usepackage{algorithm}
\usepackage{algpseudocode}

\startlocaldefs
\theoremstyle{plain}

\newtheorem{theorem}{Theorem}[section]
\newtheorem{lemma}[theorem]{Lemma}

\newtheorem{proposition}{Proposition}
\theoremstyle{remark}
\newtheorem{definition}[theorem]{Definition}

\endlocaldefs

\begin{document}

\begin{frontmatter}
\title{Markov Renewal is All You Need!}
\runtitle{Markov Renewal Proportional Hazards Model}

\begin{aug}
\author[E]{\fnms{Elvis  Han}~\snm{Cui}\thanks{\textbf{Corresponding author.}}\ead[label=e1]{elviscuihan@g.ucla.edu}},

\address[E]{Department of Neurology,  University of California, Irvine\printead[presep={, }]{e1}}

\end{aug}

\begin{abstract}
Transition probability estimation plays a critical role in multi-state modeling, especially in clinical research. This paper investigates the application of semi-Markov and Markov renewal frameworks to the EBMT dataset, focusing on six clinical states encountered during hematopoietic stem cell transplantation. By comparing Aalen-Johansen (AJ) and Dabrowska-Sun-Horowitz (DSH) estimators, we demonstrate that semi-Markov models, which incorporate sojourn times, provide a more nuanced and temporally sensitive depiction of patient trajectories compared to memoryless Markov models. The DSH estimator consistently yields smoother probability curves, particularly for transitions involving prolonged states. We use empirical process theory and Burkholder-Davis-Gundy inequality to show weak convergence of the estimator. Future work includes extending the framework to accommodate advanced covariate structures and non-Markovian dynamics.
\end{abstract}


\end{frontmatter}

\section{Review of Development of Semi-Markov Models}

In survival studies, modeling a patient's experience as a stochastic process with finite states has become a cornerstone of statistical analysis, particularly in clinical research \citep{andersen2002multi, buhler2023multistate, therneau2024using, morsomme2025assessing}. These so-called multi-state models, which rely on transition probabilities, provide a flexible framework for understanding disease progression, treatment efficacy, and survival trajectories \citep{hougaard1999multi, vasquez2024multistate}. From a probabilistic standpoint, Markov, semi-Markov, and Markov renewal processes constitute the foundational frameworks for multi-state modeling, with each introducing unique perspectives on temporal dynamics and memory. The foundational work on Markov renewal and semi-Markov processes began with \citet{pyke1961markov, pyke1961markov2}, who first introduced these concepts and rigorously formalized their properties. Pyke’s contributions established the mathematical framework for Markov renewal theory, laying the groundwork for its subsequent applications in survival analysis and clinical modeling. Notably, his work introduced essential tools to model systems where transitions depend on both the current state and the time spent in that state (sojourn time). Building on Pyke’s work, \citet{taga1963limiting} explored limiting distributions in Markov renewal processes, extending their theoretical applicability to a broader class of stochastic systems. Taga’s findings were pivotal in understanding the long-term behavior of such processes, particularly in systems characterized by recurrent events. Around the same time, \citet{ccinlar1969semi} developed a systematic study of semi-Markov processes on arbitrary spaces, emphasizing their generalization beyond Markov processes by allowing sojourn times to follow non-exponential distributions. This generalization made semi-Markov processes more suitable for modeling real-world systems where memory plays a crucial role. A particularly influential contribution came from \citet{cinlar1969markov}, who provided a comprehensive overview of Markov renewal processes. This work, later elaborated in Chapter 10 of \citet{cinlar2013introduction}, offered a detailed account of both theoretical properties and practical applications, solidifying the role of Markov renewal theory as a versatile modeling framework. Cinlar’s insights into the semi-Markov kernel and transition probabilities set the stage for future developments in estimation and inference. 

The introduction of semi-Markov processes to clinical trials was pioneered by \citet{weiss1965semi}. Their work demonstrated the utility of these models in capturing patient trajectories where the duration of time spent in each state significantly influenced transition probabilities. This innovation addressed a critical gap in traditional Markov models, which assumed memoryless transitions. Similarly, \citet{moore1968estimation} tackled the challenging problem of estimating transition distributions in Markov renewal processes, providing practical methods for parameter estimation and laying the foundation for future computational advances. In a groundbreaking dissertation, \citet{aalen1975statistical} explored semi-Markov models under the progressive assumption, where states are not revisited once left. Aalen’s contributions to counting processes and martingale theory marked a significant turning point, introducing tools that became integral to survival analysis. His work provided the theoretical underpinning for semi-Markov models, enabling their application to practical problems such as bone marrow transplant studies.

One of the most significant advancements in the field came from \citet{gill1980nonparametric}, who integrated counting processes, martingale theory, and stochastic integrals to study estimation in Markov renewal models with right-censored observations. This work not only formalized the estimation framework but also introduced powerful tools for handling incomplete data, a common issue in clinical research. Gill’s approach has since become a cornerstone for both parametric and non-parametric analysis in survival studies. Building on Gill’s foundation, \citet{gill1981testing} considers renewal testing problem under random censorship and \citet{voelkel1984nonparametric} extended martingale theory to progressive semi-Markov models, offering a nonparametric framework for inference. This extension was particularly useful for applications involving progressive illnesses, where patients transition through distinct disease states without returning to previous states. The semi-parametric Cox regression model for semi-Markov processes, introduced by \citet{dabrowska1994cox} and based on the dissertation of \citet{sun1992markov}, represented another major milestone. This model enabled researchers to incorporate covariates into semi-Markov models, significantly enhancing their flexibility and applicability to real-world data. While \citet{dabrowska1994cox} and \citep{dabrowska1995estimation, dabrowska1996nonparametric} provided a comprehensive framework for estimating transition probabilities, Dr. Dabrowska later expressed reservations about the Cox model’s assumptions in semi-Markov settings. These concerns led to the exploration of transformation models and advanced asymptotic results in her subsequent work \citep{dabrowska2006estimation,dabrowska2006estimation2,dabrowska2009estimation, dabrowska2012estimation, dabrowska2014multivariate}. These studies not only addressed limitations in existing models but also highlighted the importance of robust statistical methods in complex survival analyses. However, the lack of publicly available computational tools for implementing these models remains a barrier, which this paper aims to address.

The practical inadequacy of Markov models in certain clinical scenarios has been well-documented. For instance, \citet{andersen2000multi} demonstrated cases where the memoryless assumption failed to capture the nuanced dynamics of patient trajectories, necessitating the adoption of semi-Markov approaches. Similarly, \citet{shu2007asymptotic} established asymptotic results for a semi-Markov illness-death model, providing a rigorous theoretical foundation for analyzing progressive diseases. These contributions underscore the importance of semi-Markov models in accurately representing clinical pathways. Significant advancements have also been made in addressing practical challenges. \citet{satten1999fitting} developed methods for handling interval-censored data, while \citet{datta2002estimation} tackled the issue of dependent censoring. The introduction of nonparametric estimation techniques for non-homogeneous semi-Markov models by \citet{lucas2006nonparametric} further expanded the applicability of these models to diverse datasets. More recently, \citet{titman2010semi} introduced a novel framework using phase-type sojourn distributions, leveraging hidden Markov techniques for efficient computation. This approach has been particularly valuable for applications requiring detailed modeling of sojourn times. Comparative studies by \citet{yang2011parametric} and \citet{asanjarani2022estimation} provided critical insights into the strengths and limitations of different modeling approaches. Meanwhile, \citet{spitoni2012estimation} offered a systematic overview of estimation methods, highlighting both non-parametric and semi-parametric techniques. 

Applications of semi-Markov models in clinical research have demonstrated their versatility and impact. For example, \citet{aralis2016modeling} applied these models to dementia progression, while \citet{aralis2019stochastic} developed EM-type algorithms for estimating transition probabilities. The extension of landmark estimation frameworks to non-Markov models, including semi-Markov processes, by \citet{hoff2019landmark, maltzahn2020hybrid} has further enhanced their utility in dynamic prediction. The conditional Aalen-Johansen estimator introduced by \citet{bladt2023conditional} represents the latest advancement in this evolving field. This estimator offers a robust method for analyzing finite-state jump processes, addressing key limitations in traditional approaches. Comprehensive tutorials and software tools have also played a critical role in democratizing access to these methodologies \citep{meira2009multi}. Resources such as the \textit{mstate} package \citep{putter2007tutorial, putter2011tutorial} and monographs by \citet{cook2018multistate} and \citet{andersen2023models} provide invaluable guidance for both researchers and practitioners.

To formalize these processes, we assume observations for each individual form a Markov renewal process with a finite state space, $\{1, 2, \dots, r\}$ \citep[see][Chapter 10]{cinlar2013introduction}. Specifically, we observe a process \((X, T) = \{(X_n, T_n) : n \geq 0\}\), where \(0 = T_0 < T_1 < T_2 < \dots\) represent the times of transitions into states \(X_0, X_1, \dots, X_n \in \{1, 2, \dots, r\}\). For instance, in a bone marrow transplant (BMT) example, \(r = 6\), and \(X_n\) can take values from \{\text{TX, REC, AE, RECAE, REL, DEATH}\}. The sojourn time in state \(X_n\) is denoted by \(W_n = T_n - T_{n-1}\). Additionally, we observe a covariate matrix \(\mathbf{Z} = \{\mathbf{Z}_{ij} : i, j = 1, 2, \dots, r\}\), where each \(\mathbf{Z}_{ij}\) is itself a vector, enabling individualized modeling. Through these developments, semi-Markov models have bridged methodological rigor with clinical impact, offering tools to unravel the complexities of patient trajectories and inform precision medicine. By addressing their theoretical and computational challenges, researchers can continue to expand their applicability to diverse and impactful areas of clinical and survival research.

\section{Notations and Definitions}

In this section, we define the basic quantities used throughout the paper. Readers seeking a more comprehensive overview of inhomogeneous Markov processes and counting processes may refer to Chapter II in \citet{andersen2012statistical} and \citet{cui2022tutorial, cui2024tutorial}, while detailed discussions of Markov renewal processes can be found in Chapter 10 of \citet{cinlar2013introduction}.

\subsection{Basic Definitions and Quantities}

\begin{definition}[Markov Renewal Processes]\label{def:markov_renewal}
A \textit{Markov renewal process} is a stochastic process $\{(X_n, T_n), n \geq 0\}$, characterized as follows:
\begin{itemize}
    \item $X_n$ denotes the state of the process after the $n$-th transition, where $X_n \in \mathcal{X}$, and $\mathcal{X}$ is a finite state space with size $r$.
    \item $T_n$ represents the time of the $n$-th transition, with $T_0 = 0$ and $T_n < T_{n+1}$.
    \item $W_n = T_n - T_{n-1}$ denotes the sojourn time between the $(n-1)$-th and $n$-th transitions.
\end{itemize}
The process satisfies the Markov property, meaning that the conditional probability of the next state and transition time depends only on the current state:
\[
\mathbb{P}(X_{n+1} = j,\ T_{n+1} - T_n \leq t \mid X_0, T_0, \dots, X_n, T_n) = \mathbb{P}(X_{n+1} = j,\ T_{n+1} - T_n \leq t \mid X_n = i).
\]
This defines a joint distribution of the next state and the time until the next transition, governed by a \textbf{Markov renewal kernel} (or semi-Markov kernel).
\end{definition}

\begin{definition}[Semi-Markov Processes]
A \textit{homogeneous semi-Markov process} is a special case of a Markov renewal process where the stochastic process $\{X(t), t \geq 0\}$, tracking the state of the system over time, is defined as:
\[
X(t) = X_n, \quad T_n \leq t < T_{n+1}.
\]
\end{definition}

In a semi-Markov process, the sojourn time in each state can follow an arbitrary distribution, making it more general than a homogeneous continuous-time Markov process, where the sojourn times are exponentially distributed. Unlike Markov processes, the transitions between states are governed by the \textbf{semi-Markov kernel}, rather than transition intensities or probabilities directly:
\[
Q_{ij}(x) = \mathbb{P}(X_{n+1} = j, T_{n+1} - T_n \leq x \mid X_n = i),
\]
and we denote the semi-Markov kernel matrix as $\mathbf{Q} = [Q_{ij}]$. This kernel defines the probability of transitioning from state $i$ to state $j$ within time $x$. Importantly, we assume no self-transitions, i.e., $Q_{ii}(x) = 0$ for all $i$ and $x$.

\begin{definition}[Distribution and Survival Probabilities in State $i$]\label{def:G_H}
We define the following quantities, which are fundamental to subsequent developments:
\begin{align}
H_i(x) &= \sum_{k=1}^r Q_{ik}(x) = \mathbb{P}(W_n \leq x \mid X_{n-1} = i), \\
G_i(x) &= 1 - H_i(x)=\mathbb{P}(W_n>x|X_{n-1}=i).
\end{align}
\end{definition}

\begin{definition}[Transition Probabilities]
The \textit{transition probability matrix} $\mathbf{P}$ of a semi-Markov process $X$ is given by:
\begin{align}
\mathbf{P}(s, t) &= [P_{ij}(s, t)], \\
P_{ij}(s, t) &= \mathbb{P}(X(t) = j \mid X(s) = i).
\end{align}
\end{definition}

\textbf{Remark}: The semi-group property of Markov processes does not hold here because transitions depend not only on the current state but also on the duration spent in that state. For simplicity, we write $\mathbf{P}(t)$ and $P_{ij}(t)$ when $s = 0$, i.e., when the process starts at $0$ in state $i$.

\begin{definition}[Counting Process Formulation]\label{def:risk}
Let $(X, T)$ be a Markov renewal process as defined in Definition~\ref{def:markov_renewal}. We adopt the random censorship framework in \citet{gill1980nonparametric} and \citet{dabrowska1994cox}, where observation times are determined by a predictable 0-1 process $K(t) = \sum_{n=1}^\infty \mathbb{I}(T_n < t \leq C_n)$, with $T_n$ representing the $n$-th jump time and $C_n \in [T_n, T_{n+1}]$ denoting censoring variables. The observed number of $i \to j$ transitions before sojourn time $x$ is:
\[
N_{ij}(x) = \sum_{n=1}^\infty \mathbb{I}(X_n = j, X_{n-1} = i, W_n \leq x, K(T_n) = 1).
\]
The risk process associated with state $i$ is:
\[
Y_i(x) = \sum_{n=1}^\infty \mathbb{I}(X_{n-1} = i, W_n \geq x, K(x + T_{n-1}) = 1),
\]
which counts the number of observed sojourn times in state $i$ that are at least $x$.
\end{definition}

\begin{definition}[Transition-Specific Intensities]
Given $(X, T)$ and $W$, the \textit{transition-specific intensity} (or hazard) is defined as:
\[
\alpha_{ij}(x) = \lim_{\Delta x \downarrow 0} \frac{1}{\Delta x} \mathbb{P}(X_n = j, W_n \leq x + \Delta x \mid X_{n-1} = i, W_n \geq x).
\]
The matrix of transition intensities $\mathbf{A} = [A_{ij}]$ is defined as:
\[
A_{ij}(t) = \int_0^t \alpha_{ij}(u) du,
\]
with $A_{ii}(t) = 0$ for all $i$ and $t$. For simplicity, we assume the existence of $\alpha_{ij}$ such that $A_{ij}$ is absolutely continuous.
\end{definition}

\textbf{Remarks}: In Markov processes (homogeneous or inhomogeneous), the diagonal elements of $\mathbf{A}$ satisfy $A_{ii}(t) = -\sum_{j \neq i} A_{ij}(t)$. The term $\lambda_{ij}(t) = Y_i(t) \alpha_{ij}(t)$ is recognized as the intensity process of $N_{ij}(t)$, where $t$ refers to calendar time. Importantly, the likelihood function of a semi-Markov process closely resembles that of a Markov process by Jacod's formula (see Pages 680--681 in \citet{andersen2012statistical}). If the semi-Markov process is progressive (i.e., states are not revisited), we can replace the condition $W_n \geq t$ with the calendar time $T_n \leq t$, enabling the application of Andersen-Gill's counting process martingale framework. This transformation, referred to as Aalen's random time transformation, is widely discussed in the literature \citep{aalen1978nonparametric}.

\subsection{Transition Probabilities}

\begin{lemma} Let $\mathbf{B}$ be a diagonal matrix with elements $B_{i}=\sum_{j\not=i}A_{ij}$. The relation between $\mathbf{A}$ and $\mathbf{G}=[G_i]$ is given by the following product integral \citep{gill1990survey}:
    \begin{align}\label{eq:survival}
        \mathbf{G}(x)&=\Prodi_{u\in(0,x]}\left( \mathbf{I}-\mathbf{B}(du) \right)
    \end{align}
    where $\mathbf{I}$ is an $r\times r$ indentity matrix.
\end{lemma}
\begin{proof}
    It is enough to show that
    $$G_i(x)=1-H_i(x)=\mathbb{P}\left(W_n> x|X_{n-1}=i\right)=\Prodi_{u\in(0,x]}\left(1-\sum_{j\not=i}A_{ij}(du)\right).$$
    But $\sum_{j\not=i}A_{ij}(u)$ is the cumulative hazard of staying at state $i$. In the case that $A_{ij}$ are absolutely continuous, we have
    $$G_i(x)=\exp(-\sum_{j\not=i}A_{ij}(x)).$$
\end{proof}
\begin{lemma}
    The semi-Markov kernel matrix $\mathbf{Q}$ is equal to
    \begin{align}
        \mathbf{Q}(x)&=\int_0^x\mathbf{G}(u-)\mathbf{A}(du).
    \end{align}
\end{lemma} 
\begin{proof}
    We have
    \begin{align*}
        \int_0^xG_i(u-)A_{ij}(du)&=\int_0^x\mathbb{P}(W_n\ge u|X_{n-1}=i)\mathbb{P}(X_n=j,W_n\in[u,u+du]|X_{n-1}=i,W_n\ge u)\\
        &=\mathbb{P}(W_n\le x,X_n=j|X_{n-1}=j)\\
        &=Q_{ij}(x).
    \end{align*}
\end{proof}
This lemma suggests a plug-in estimator for $\mathbf{Q}$ \citep{gill1980nonparametric}:
$$\widehat{\mathbf{Q}}(x)=\int_0^x\widehat{\mathbf{G}}(u-)\widehat{\mathbf{A}}(du)$$
where $\widehat{G}(u),\ \widehat{A}(u)$ are Kaplan-Meier and Nelson-Aalen estimators, respectively (See section~\ref{sec:est}). However, transition probabilities are significantly more challenging to handle compared to the Markov case, where the elegant Aalen-Johansen estimator provides a straightforward solution. For semi-Markov and Markov renewal processes, we have the following backward equation:
$$\mathbf{P}(t)=\mathbf{G}(t)+\int_0^t\mathbf{Q}(du)\mathbf{P}(t-u).$$
Here we are using the calendar time notation $t$ instead of the sojourn time notation $x$ because the system starts at time $0$ so that the first sojourn time aligns with the calendar time. The equation can be written more compactly in a renewal form:
\begin{align}\label{eq:markov_renewal}
    \mathbf{P}(t)=\mathbf{G}(t)+\mathbf{Q}*\mathbf{P}(t)
\end{align}
where $*$ denotes the matrix convolution operator \citep{cinlar2013introduction}.
\begin{definition}[Markov renewal function]
    The \textit{Markov renewal function} of a Markov renewal process $(X,T)$ is the matrix $\mathbf{R}=[R_{ij}]$ where
    \begin{align}
        R_{ij}(t)&=\mathbb{E}\left(\sum_{n=0}^\infty \mathbb{I}(X_n=j,T_n\le t)\Big|X_0=i\right)\nonumber\\
        &=\sum_{n=0}^\infty \mathbf{Q}^{(n)}(t)
    \end{align}
    where $\mathbf{Q}^{(0)}=\mathbf{I}$ and 
    $$\mathbf{Q}^{(n)}(t)=\int_0^t\mathbf{Q}(du)\mathbf{Q}^{(n-1)}(t-u)=\mathbf{Q}*\mathbf{Q}^{(n-1)}(t).$$
\end{definition}
\begin{theorem}[Transition probabilities \citep{spitoni2012estimation, cinlar2013introduction}]\label{thm:trans_prob}
    The unique solution to the Markov renewal equation \ref{eq:markov_renewal} is
    \begin{align}
        \mathbf{P}(t)=\mathbf{R}*\mathbf{G}(t)=\int_0^t\mathbf{R}(du)\mathbf{G}(t-u).
    \end{align}
\end{theorem}
The proof is omitted and can be found in page 324 in \citet{cinlar2013introduction}. The intuition behind $\mathbf{Q}^{(n)}$ is that it defines $n$-step transition semi-Markov kernel. We provide a summary of the concepts defined so far in Table~\ref{tab:markov_vs_semimarkov}.

\section{Estimation in Markov Renewal Cox Regression Model}\label{sec:est} 
\subsection{Likelihood Function}
Let $\mathbf{Z}=[Z_{ij}]$ be an $r\times r$ transition specific covariate matrix so that each element $Z_{ij}$ is a time-independent vector associated with the  transition $i\rightarrow j$. We assume the transition intensity (or transition hazard) follows a Cox regression model \citep{dabrowska1994cox, dabrowska1995estimation}:
\begin{align}\label{eq:regression}
    \alpha_{ij}(x)=\alpha_{0ij}(x)e^{\beta^TZ_{ij}}
\end{align}
where $x=t-T_n$ is the sojourn time spent in state $i$ and $\alpha_{0ij}$ is the baseline hazard as in the usual Cox regression models. The goal of this section is to provide estimators for the regression coefficient $\beta$ and transition probabilities $\mathbf{P}$. We start with the likelihood function.
\begin{theorem}[Jacod's formula for likelihood functions \citep{sun1992markov, andersen2012statistical}] The likelihood for observing a Markov renewal processes $(X,T)$ is
\begin{align}
        \prod_{n \geq 1} 
\alpha_{X_{n-1}, X_n}(W_n)
\times 
\exp\left[
-\sum_{n \geq 1} \sum_k 
\int_{0}^{W_n} 
\alpha_{X_{n-1}, k}(u) \, du 
\right].
    \end{align}
    In terms of calendar time $T$, we write
    \begin{align*}
\prod_{n \geq 1} \prod_{i\not=j}
\alpha_{ij}(T_n-T_{n-1})^{\mathbb{I}(X_{n-1}=i,X_n=j)}
\times 
\exp\left[
-\sum_{n\ge 1}\sum_j 
\int_{T_{n-1}}^{T_n} 
\alpha_{ij}(t-T_{n-1}) \, dt 
\right].
    \end{align*}
\end{theorem}

\textbf{Remark}:     When there are $M$ subjects or $M$ independent $(X,T)$'s, then we just simply multiply them together. In the sequel, we assume that $M$ independent copies of $(X,T,Z)$ is observed and we use $Y_{i},N_{ij},$ etc. to denote the summation of these processes, dropping the dependence on $m$. For example, $N_{ij}(x)=\sum_mN_{mij}(x)$ counts the total number of transition $i\rightarrow j$ registered by all subjects. Finally, we define $N_i(x)=\sum_{j=1}^rN_{ij}(x)$ to be the total number of transitions from $i$.

\subsection{Non- and Semi-parametric Estimation}
By a heuristic argument (Chapter 2 in \citet{andersen2012statistical}), one can estimate the cumulative intensity $A_{ij}$ by
\begin{align}
    \widehat{A}_{ij}(x)=\int_0^x\frac{J_i(u)dN_{ij}(u)}{Y_{i}(u)}
\end{align}
where $Y_{i}(x)$ is the risk process in Definition~\ref{def:risk} and $J_i(u)=\mathbb{I}(Y_{i}(u)>0)$. If we assume the regression model~\ref{eq:regression}, then a Breslow-type estimator (assuming that $\beta$ is known) for the baseline hazard is suggested by \citet{sun1992markov}:
\begin{align}
    \widehat{A}_{0ij}(x)=\int_0^x\frac{\mathbb{I}(S_{ij}^{(0)}(u,\beta)>0)}{mS_{ij}^{(0)}(u,\beta)}N_{ij}(du)
\end{align}
where
$$S_{ij}^{(0)}(x,\beta)=m^{-1}\sum_mY_{mi}(x)\exp\left( \beta^TZ_{mij} \right).$$
\citet{gill1980nonparametric} suggested the following non-parametric estimator for the semi-Markov kernel:
$$\widehat{\mathbf{Q}}(x)=\int_0^x\widehat{\mathbf{G}}(u-)\widehat{\mathbf{A}}(du)$$
where the diagonal elements of $\mathbf{\widehat{A}}$ are $0$ and diagonal elements of $\widehat{\mathbf{G}}$ are
$$\widehat{G}_i(x)=\Prodi_{u\in(0,x]}\left(1-\widehat{B}_i(du)\right),\  \widehat{B}_i(u)=\sum_{j\not=i}\widehat{A}_{ij}(u).$$
\textbf{Remark}: The estimators $ \widehat{G}_i,\widehat{A}_{ij}$ correspond to Kaplan-Meier and Nelson-Aalen in the classical case. When there is only one possible progressive state starting from $i$, i.e., $i\rightarrow j$ and $j$ is absorbing, then the $i$-th row of $\widehat{\mathbf{Q}}$ corresponds to the estimation of the cause-specific distribution function in the language of competing risks data.

In the semi-parametric setting (Model~\ref{eq:regression}), \citet{sun1992markov} and \citet{dabrowska1994cox}  derived the estimating equation for $\beta$ by plug-in the Breslow-type estimator $\widehat{A}_{0ij}$ back into the likelihood function and then taking logarithm. The resulting profile log-likelihood \citep{johansen1983extension} is equal to $C(\infty,\beta)$ where
\begin{align}    C(\tau,\beta)&=\sum\limits_{m}\sum\limits_{i\not=j}\int_0^\tau\left[ \beta^TZ_{mij} - \log\left( mS_{ij}^{(0)}(x,\beta) \right) \right]N_{mij}(dx)
\end{align}
where we unavoidably use the subscript $m$ to indicate the $m$-th subject. The estimator for $\beta$ is derived by solving the corresponding estimating equation. Asymptotics for both $\beta$ and $\mathbf{A}$ are derived in \citet{sun1992markov, dabrowska1994cox, spitoni2012estimation, shu2007asymptotic} using  lengthy arguments. The semi-Markov kernel $\mathbf{Q}$ is estimated similarly in the non-parametric case. Explicitly, we have
$$\widehat{Q}_{ij}(t)=\int_0^t\widehat{G}_i(u-)\widehat{A}_{ij}(du)$$
where
$$\widehat{A}_{ij}(du)=\widehat{A}_{0ij}(du)\sum_me^{\widehat{\beta}^TZ_{mij}}.$$

\subsection{Estimation of Transition Probabilities}

By theorem~\ref{thm:trans_prob}, the estimator for transition probability matrix $\mathbf{P}$ is
\begin{align}\label{eq:est_markov_renewal}
    \widehat{\mathbf{P}}(t)&=\widehat{\mathbf{R}}*\widehat{\mathbf{G}}(t)=\int_0^t\widehat{\mathbf{R}}(du)\widehat{\mathbf{G}}(t-u)
\end{align}
where
$$\widehat{\mathbf{R}}=\sum_{p=0}^\infty\widehat{\mathbf{Q}}^{(p)},\ \widehat{\mathbf{Q}}^{(p)}(t)=\int_0^t\widehat{\mathbf{Q}}(du)\widehat{\mathbf{Q}}^{(p-1)}(t-u).$$
It is computationally intractable if the process is not progressive, i.e., $\mathbf{P}$ has non-zero lower-triangular elements. But the convolution only involves finite  terms if the process is progressive, see section 3.3.2 in \citet{dabrowska1994cox}.

Finally, table~\ref{tab:markov_vs_semimarkov} provides a detailed comparison of semi-Markov and Markov processes, highlighting their key differences and similarities across various dimensions. The table covers essential aspects such as notations, time scales, intensity matrices, transition probabilities, counting processes, and likelihood formulations. While semi-Markov processes operate primarily on the sojourn time scale, Markov processes are defined on the calendar time scale. Despite this fundamental distinction, it is possible to analyze semi-Markov processes in calendar time by appropriately transforming the time variables. This comparison is intended to clarify the theoretical underpinnings and computational considerations for modeling and estimation in these two process frameworks.

\begin{table}[ht]
    \centering
    \renewcommand{\arraystretch}{1.5} 
    \begin{threeparttable}
    \begin{tabular}{@{}p{3.5cm} p{6.5cm} p{6.5cm}@{}}
    \toprule
       & \textbf{Semi-Markov processes}  & \textbf{Markov processes}   \\ 
    \midrule
    \textbf{Notations} & $(X,T)$ and $W$ for Markov renewal representation; $\{X(t):t\ge 0\}$ for semi-Markov. The sequence $(X_n,W_n)$ forms a homogeneous Markov chain.& Same representation but the future $X(t)$ only depends on the current value $X(s)$, not the past $X(u),u<s$. \\
       \textbf{Time scale}\textsuperscript{1} & 
       \(x\): The sojourn time between jumps. & 
       \(t\): The calendar time. \\ 

       \textbf{Intensity matrix}\textsuperscript{2} & 
       \(\mathbf{A}(x) = [A_{ij}(x)],\ A_{ii} = 0\) (for progressive processes). & 
       \(\mathbf{A}(t) = [A_{ij}(t)],\ A_{ii} = -\sum_{j \neq i} A_{ij}\). \\ 
       \textbf{Sojourn time at state $i$} &$\Prodi_{u\in(0,x]}\left(1-\sum_{j\not=i}A_{ij}(du)\right)$ &$\Prodi_{u\in(s,t]}\left(1+A_{ii}(du)\right)$  where $s$ is the entry time at state $i$. It reduces to the homogeneous Markov case if $A_{ij}(u)=q_{ij}u,q_{ij}\ge0$. \\
        \textbf{One step transition} &$\mathbf{Q}$: the semi-Markov kernel. &No direct semi-Markov kernel analog. \\
       \textbf{Transition probability}&$\mathbf{P}(t)=\mathbf{R}*\mathbf{G}(t)$ where $\mathbf{R}$ is the Markov renewal function and $\mathbf{G}$ is the diagonal matrix defined in \ref{eq:survival}. & $\mathbf{P}(s,t)=\Prodi_{u\in(s,t]}\left(\mathbf{I}+\mathbf{A}(du)\right)$\\

       \textbf{Counting processes}& {\scriptsize $N_{ij}(x)=\sum_{n=1}^\infty\mathbb{I}(T_n-T_{n-1}\le x,X_{n}=j,X_{n-1}=i)$} &{\scriptsize $N_{ij}(t)=\sum_{n=1}^\infty\mathbb{I}(T_n\le t,X_{n}=j,X_{n-1}=i)$}\\
       \textbf{Counting process}
       
       \textbf{martingales}& No filtration can make $N_{ij}(x)-Y_i(x)A_{ij}(x)$ a martingale where $Y_i(x)=$number of durations in $i$ with sojourn time at least $x$. & $M_{ij}(t)=N_{ij}(t)-Y_i(t)A_{ij}(t)$ is a martingale where $Y_i(t)=\mathbb{I}(X(t-)=i)$. \\
       \textbf{Likelihood}\textsuperscript{3} & 
       ${\scriptsize \mathcal{L} =\prodi_x\prod_{j\not=i}dA_{ij}^{dN_{ij}}(1-dB_i)^{Y_i-\sum_jdN_{ij}}}$ where $B_i(x)=\sum_{j\not=i}A_{ij}(x)$.  & 
       {$\mathcal{L} = \prodi_t\prod_{j\not=i}dA_{ij}^{dN_{ij}}(1+dA_{ii})^{Y_i-\sum_jdN_{ij}}$}\\
       \textbf{Implementation in $R$}& Not available, need to scratch from hand using \textit{survival} and \textit{timereg} packages \citep{scheike2011analyzing, therneau2015package}. & \textit{mstate}, \textit{msm} packages \citep{jackson2011multi, de2011mstate}.\\
    \bottomrule
    \end{tabular}
    \begin{tablenotes}
        \footnotesize
        \item[1)] Of course one can work with calendar time within the semi-Markov framework. But transition intensities and state survival probabilities are derived in terms of sojourn times.
        \item[2)] The intensity matrix \(\mathbf{A}\) specifies the transition rates. In Markov processes, it is also known as the generator.
        \item[3)] The difference here is that the likelihood for semi-Markov is defined in sojourn time scale while Markov is in calendar time scale. But one can still work with calendar time scale for the semi-Markov case with $x=t-T_n$.
        \item
    \end{tablenotes}
    \caption{Comparison of Semi-Markov and Markov Processes}
    \label{tab:markov_vs_semimarkov}
    \end{threeparttable}
\end{table}

\subsection{Prediction Probabilities}
To compute the transition probability of a semi-Markov process starting from an arbitrary calendar time $s$ to $t$, often termed as prediction probabilities, we rely on the following mathematical formula.

\begin{lemma}
Define $\mathbf{P}(s,t)=[P_{ij}(s,t)]$ to be the transition probability matrix of a semi-Markov process starting at time $s$, i.e.,
    \begin{align*}
        P_{ij}(s,t)=\mathbb{P}\left(X(t)=j|X(s)=i,\mathscr{F}_s\right)
    \end{align*}
    where $\mathscr{F}_s$ is the history up to time $s$. In the non-parametric setting it could be $\sigma(X(u):u\le s)\vee \mathscr{F}_0$, the self-exciting filtration; in the semi-parametric setting, it is enlarged by covariates. The history $\mathscr{F}_s$ then contains information on the times and states visited $(X_0,T_1),\cdots,(X_n,T_n)$. The relation between $\mathbf{P}(s,t)$ the transition probability matrix $\mathbf{P}(t)$ starting at $0$ is
    \begin{align}
        P_{ij}(s,t)&=\delta_{ij}G_i(s,t)+\sum_{k\not=i}^r\int_s^tH_{ik}(s,du)P_{kj}(t-u)\\
        G_{i}(s,t)&=\frac{G_i(t-T_n)}{G_i(s-T_n)}=\Prodi_{s-T_n}^{t-T_n}(1-B_i(du))
    \end{align}
    where
    $$H_{ik}(s,du)=\frac{Q_{ik}(d(u-T_{n}))}{G_{i}(s-T_{n})}$$
    and $Q_{ik}$ is the semi-Markov kernel of $X(s)$, $B_i$ and $G_{i}$ are defined in \ref{def:G_H}, the survival function at state $i$. In matrix notation, we have
    \begin{align}
        \mathbf{P}(s,t)&=\mathbf{G}(s,t)+\int_s^t\mathbf{H}(s,du)\mathbf{P}(t-u)
    \end{align}
    where $\mathbf{G}(s,t)$ is a diagonal matrix and $\mathbf{H}$ has $0$ on the diagonal  and other entries are given by $H_{ik}$.

\end{lemma}
\begin{proof}
    We have
    \begin{align*}
        \mathbb{P}\left( X(t)=j|X(s)=i,\mathscr{F}_s \right)=&\ \text{Case 1} + \text{Case 2}
    \end{align*}
    where
    \begin{align*}
        \text{Case 1 }&=\mathbb{P}(X(u)=j,\forall u\in(s,t]|X(s)=i,\mathscr{F}_s)\\
        &=\delta_{ij}\mathbb{P}(T_{n+1}\ge t|X_n=i,T_{n+1}\ge s,T_n)\\
        &=\delta_{ij}\Prodi_{s-T_n}^{t-T_n}(1-B_i(du))\\
        &=\delta_{ij}G_i(s,t)
    \end{align*}
    and
    \begin{align*}
         \text{Case 2}=&\int_s^t\sum_{k\not=i}\mathbb{P}(X(t)=j|T_{n+1}=u,X_{n+1}=k,X_n=i,\mathscr{F}_s)\times\\
        &\ \ \ \ \ \ \mathbb{P}(T_{n+1}\in[u,u+du],X_{n+1}=k|X_n=i,\mathscr{F}_s)\\
        &=\sum_{k\not=i}\int_s^tP_{kj}(t-u)\frac{Q_{ik}(d(u-T_n))}{G_i(s-T_n)}
    \end{align*}
    since
    \begin{align*}
        \mathbb{P}(T_{n+1}\le u,X_{n+1}=k|X_n=i,\mathscr{F}_s)&=\mathbb{P}\left(W_{n+1}\le u-T_n,X_{n+1}=k|W_{n+1}\ge s-T_n,X_n=i\right)\\
        &=\frac{Q_{ik}(u-T_n)-Q_{ik}(s-T_n)}{1-\sum_kQ_{ik}(s-T_n)}
    \end{align*}
\end{proof}
\textbf{Remark}: A different approach for tackling Case 2 is given in Section 2.3 of \citet{sun1992markov}. Section 3.2.1 in \citet{dabrowska1994cox} gives a somewhat similar formula but more difficult to implement. The lemma is non-trivial in the sense that $P_{ij}(s,t)\not= P_{ij}(t)/P_{ij}(s)$ in contrast to the Markov case. Using the lemma, we estimate $\mathbf{P}(s,t)$ by
$$\widehat{\mathbf{P}}(s,t)=\widehat{\mathbf{G}}(s,t)+\int_s^t\widehat{\mathbf{H}}(s,du)\mathbf{\widehat{P}}(t-u)$$
where $\mathbf{\widehat{G}}(s,t)$ is obtained by plug-in estimates of $\widehat{G}_i$, $\mathbf{\widehat{H}}$ is obtained by plug-in $\widehat{\mathbf{Q}}$ and $\widehat{\mathbf{P}}$ is given by Formula~\ref{eq:est_markov_renewal}. We refer to $\mathbf{\widehat{P}}$ as the Dabrowska-Sun-Horowitz (DSH) estimator.

\section{Inference in Markov Renewal Cox Regression Model}\label{sec:inference}
\subsection{Some Asymptotics}
The inference for the regression coefficient $\widehat{\beta}$ and the cumulative hazard are done via Law of Large Numbers (LLN) and empirical process methods \citep{shorack2009empirical} instead of martingale arguments. The results can be used to construct confidence intervals as well as confidence bands via Reeds-Gill's theorem (the functional delta method) \citep{gill1989non} and Efron's bootstrapping.

\begin{definition}[Some additional notations]
    Recall that    $$C(\tau,\beta)=\sum\limits_{m}\sum\limits_{i\not=j}\int_0^\tau\left[ \beta^TZ_{mij} - \log\left( mS_{ij}^{(0)}(x,\beta) \right) \right]N_{mij}(dx),$$
    and for $i,j\le r$, we denote $S_{ij}^{(1)}$ and $S_{ij}^{(2)}$ the vector and respectively the matrix of the first and second partial derivatives of $S_{ij}^{(0)}$ with respect to $\beta$. We set 
    $E_{ij}(x, \beta) = \frac{S^{(1)}_{ij}(x, \beta)}{S^{(0)}_{ij}(x, \beta)} \quad \text{and} \quad V_{ij}(x, \beta) = \frac{S^{(2)}_{ij}(x, \beta)}{S^{(0)}_{ij}(x, \beta)} - \big(E_{ij}(x, \beta)\big)^{\otimes 2}$, and then the corresponding estimating equation (or score equation) is
    $$\widehat{U}(\tau,\beta)=\sum_m\sum_{i\not=j}\int_0^\tau \left( Z_{mij}(x)-E_{ij}(x,\beta) \right)N_{mij}(dx).$$
    Finally, similar to \citet{andersen1982cox}, we set
    \begin{align}
        s_{ij}^{(p)}(x,\beta)&=\mathbb{E} S_{ij}^{(p)}(x,\beta),\ p=0,1,2,\nonumber\\
        e_{ij}(x,\beta)&=\frac{s^{(1)}_{ij}(x,\beta)}{s^{(0)}_{ij}(x,\beta)},\nonumber\\
        v_{ij}(x,\beta)&=\frac{s_{ij}^{(2)}(x,\beta)}{s_{ij}^{(0)}(x,\beta)}-e_{ij}(x,\beta)^{\otimes 2},\nonumber\\
        \Sigma(x,\beta)&=\sum_{i,j\le r}\int_0^xv_{ij}(u,\beta)s_{ij}^{(0)}(u,\beta)\alpha_{0,ij}(u)du.
    \end{align}

\end{definition}

\begin{theorem}[Weak convergence of $\widehat{\beta}$ and $\widehat{A}$ \citep{dabrowska1994cox}]\label{thm:weak_conv}
 Under the regularity conditions in Appendix~\ref{sec:regularity}, for the Markov renewal Cox regression model~\ref{eq:regression}, we have:
\[
\sqrt{m}(\widehat{\beta} - \beta_0) \xrightarrow{d} \Sigma^{-1}(\tau, \beta_0) U(\tau, \beta_0),
\]
and
\[
\sqrt{m} \big(\widehat{A}_{0ij}(x, \widehat{\beta}) - A_{0ij}(x, \beta_0)\big) \xrightarrow{d} \Psi_{ij}(x, \beta_0) + \eta_{ij}(x, \beta_0)^\top \Sigma^{-1}(\tau, \beta_0) U(\tau, \beta_0),
\]
where $\beta_0$ is the true parameter, \(U(x, \beta_0)\) and \(\Psi_{ij}(x, \beta_0)\) are mean-zero Gaussian processes with independent components and covariances:
\[
\text{Cov}[U(x, \beta_0), U(y, \beta_0)] = \Sigma(x \wedge y, \beta_0),
\quad
\text{Cov}[\Psi_{ij}(x, \beta_0), \Psi_{ij}(y, \beta_0)] = \gamma_{ij}(x \wedge y, \beta_0).
\]
Here,
\[
\eta_{ij}(x, \beta) = - \int_{0}^{x} e_{ij}(u, \beta) \alpha_{0,ij}(u) \, du,
\quad
\gamma_{ij}(x, \beta_0) = \int_{0}^{x} \big[s^{(0)}_{ij}(u, \beta_0)\big]^{-1} \alpha_{0,ij}(u) \, du.
\]
The weak convergence is in \(D([0, \tau])^{(r^2-r) \times (1 + d)}\) with the Skorohod topology, assuming the covariate vector \(\mathbf{Z} = \{Z_{ij}\}\) is of dimension \(d\).
\end{theorem}

\begin{proof}
    See Appendix~\ref{sec:proof_weak_conv}. Here we use $\widehat{A}_{0ij}(x,\beta)$ instead of $\widehat{A}_{0ij}(x)$ to emphasize the dependence on $\beta$, which is essential in the proof.
\end{proof}

\begin{theorem}[Consistency and Weak Convergence of $\widehat{\mathbf{Q}},\widehat{\mathbf{R}}$ and $\widehat{\mathbf{P}}$ \citep{dabrowska1995estimation, spitoni2012estimation}]

Let \( \tau \) be a point such that \( E Y_i(\tau) > 0, i = 1, \dots, r \) and suppose that \( E Y_i(0)^3 < \infty \) and the semi-Markov kernel \( \mathbf{Q} \) is continuous.  

\begin{itemize}
    \item[(i)] The one step transition process
    \begin{align}
        \sqrt{m} \big(\mathbf{\widehat{Q}} - \mathbf{Q}\big)
    \end{align}
    converges weakly in \( D[0, \tau]^{r^2} \) to a mean-zero Gaussian process \( \mathbf{\Phi} \), say. The covariance of $\mathbf{\Phi}$ is given in Appendix B-8 in \citet{spitoni2012estimation}.
    
    \item[(ii)] We have \( \widehat{\mathbf{R}}(t) \xrightarrow{p} \mathbf{R}(t) \)  uniformly in \( t \in [0, \tau] \). In addition,
    \[
    \sqrt{m} (\widehat{\mathbf{R}} - \mathbf{R}) 
    \]
    converges weakly in \( D[0, \tau]^{r^2} \) to a mean-zero Gaussian process
    \begin{align}
        \sum_p\sum_{l=1}^p\mathbf{Q}^{(p-l)}*\mathbf{\Phi}*\mathbf{Q}^{(l-1)}.
    \end{align}
    
    \item[(iii)] We have \( \mathbf{\widehat{P}}(t) \xrightarrow{p} \mathbf{P}(t) \) uniformly in \( t \in [0, \tau] \). In addition,
    $$\sqrt{m} (\widehat{\mathbf{P}} - \mathbf{P})$$
    converges weakly in $D[0,\tau]^{r^2}$ to a mean-zero Gaussian process
    $$\sum_p\sum_{l=1}^p\mathbf{Q}^{(p-l)}*\mathbf{\Phi}*\mathbf{Q}^{(l-1)}*\mathbf{G}-\mathbf{R}*\text{Diag}\left(\mathbf{\Phi}\mathbf{1} \right)$$
    where $\mathbf{1}$ is an $r$-dimensional vector of $1$'s.
\end{itemize}
\end{theorem} 
\begin{proof}
    To prove this theorem, one verifies that the mappings $\mathbf{A}\rightarrow\mathbf{Q},\mathbf{R},\mathbf{P}$ are Hadamard differentiable on the class of cadlag matrices valued functions of bounded variation on $[0,\tau]$. Then one uses the functional delta-method applied to $\sqrt{m}\left(\mathbf{\widehat{A}}-\mathbf{A}\right)$ \citep{gill1989non}.
\end{proof}

\subsection{Bootstrapping}
\citet{dabrowska1995estimation} suggests a bootstrapping procedure for the construction of confidence bands and we summarize it below. We suppose now that the underlying Markov renewal process is hierarchical and the extension to non-hierarchical models is easy. The bootstrap sample 
\[
[(T_k^*, X_k^*) = (T_{k,n}^*, X_{k,n}^*)_{n \geq 0}, Z_k^*, \tilde{T}_k^* : k = 1, \dots, m]
\]
can then be generated as follows:

\begin{enumerate}
    \item \textbf{Sampling Covariates and Initial States}:  
    For each individual \( k \in \{1, \dots, m\} \), the covariate vector \( Z_k^* \) and the initial state \( X_{k,0}^* \) are sampled with replacement from the observed covariates \( Z_p \) and initial states \( X_{p,0} \), where \( p \in \{1, \dots, m\} \). This ensures that each sampled covariate and initial state represents a realistic starting condition based on the observed data.
    
    \item \textbf{Generating Censoring Times}:  
    For each individual \( k \), a censoring time \(\widetilde{T}_k^*\) is generated as an independent random variable with survival function, say \({\mathbf{S}}(t)\). This step models the censoring mechanism, which accounts for the incomplete observation of event times due to study design or external factors. \citet{dabrowska1995estimation} gives a more explicit expression for $\mathbf{S}(t)$ based on empirical data.
    
    \item \textbf{Constructing the Sequence of Transition Times and States}:  
    Given the triple \((X_{k,0}^*, Z_k^*, \widetilde{T}_k^*)\), we iteratively construct the sequence 
    \[
    (T_{k,0}^* = 0, X_{k,0}^*), \dots, (T_{k,n}^*, X_{k,n}^*),
    \]
    as follows:
    \begin{enumerate}
        \item[(i)] If the current state is \( X_{k,n-1}^* = i \in \{1, \dots, r\} \), the next transition time is determined by 
        \[
        T_{k,n}^* = T_{k,n-1}^* + W_{k,n}^*, 
        \]
        where \( W_{k,n}^* \) is a random variable drawn from the survival function
        \[
        G_{i}(x; Z_k^*) = \Prodi_{u\in(0,x]} \left(1 - \widehat{B}_{i}(du; Z_k^*)\right),
        \]
        and \(\widehat{B}_{i}(x; Z_k^*)\) is defined as
        \[
        \widehat{B}_{i}(x; Z_k^*) = \sum_{\ell\not=i}^r e^{\beta^\top Z_{k,i\ell}^*}\widehat{A}_{0i\ell}(x).
        \]
        
        \item[(ii)] If the next transition time \( T_{k,n}^* \) exceeds the censoring time \(\widetilde{T}_k^*\), the process is censored at \(\tilde{T}_k^*\). Otherwise, the next state \( X_{k,n}^* \) is sampled from a distribution where the probability of transitioning to state \( j \in \{1, \dots, r\} \), $j\not=i$ is proportional to
        \[
        \frac{A_{ij}(\Delta W_{k,n}^*; Z_k^*)}{\widehat{B}_{i}(\Delta W_{k,n}^*, Z_k^*)}.
        \]
        
        \item[(iii)] Steps (i) and (ii) are repeated until the process either reaches an absorbing state or is censored at \(\widetilde{T}_k^*\).
    \end{enumerate}
\end{enumerate}

Given the bootstrap sample 
\[
[(T_k^*, X_k^*), Z_k^*, \widetilde{T}_k^*: k = 1, \dots, m],
\]
we estimate the regression parameter \(\beta\) by solving the equation
\[
\widehat{U}^*(\tau, \beta) = 0,
\]
where
\[
\widehat{U}^*(\tau, \beta) = \sum_{k=1}^m \sum_{ij} \int_0^\tau \big[Z_{ijk}^*(x) - E_{ij}^*(x, \beta)\big] dN_{ijk}^*(x).
\]

Finally, the bootstrap analogue of the baseline cumulative hazard function estimator is defined as \(\widehat{A}_0^*(x) = [\widehat{A}_{0;ij}^*(x)]\), where
\[
\widehat{A}_{0;ij}^*(x) = \int_0^x \frac{\mathbb{I}(S_{ij}^{(0)}(u, \beta) > 0)}{m S_{ij}^{(0),*}(u, \beta)} \sum_{k=1}^m dN_{ijk}^*(u).
\]
All ``$*$"-indexed processes have the same meaning as the none $*$ ones but are computed from the bootstrapping data $\left[\left( T_k^*,X_k^* \right),Z_k^*,\widetilde{T}_k^*,k=1,\cdots,m\right]$. The following proposition characterizes the limiting behavior of the bootstrap estimator, analogous to Theorem~\ref{thm:weak_conv}.
\begin{proposition}
    Under some regularity conditions, the bootstrap is weakly consistent, i.e.,  
\[
\sqrt{m}(\hat{\beta}^* - \hat{\beta}) \mid \text{data} \xrightarrow{d} \Sigma^{-1}(\tau, \beta_0) U(\tau, \beta_0)
\]
and
\[
\sqrt{m}(\hat{A}_{ij}^*(x, \hat{\beta}^*) - \hat{A}_{ij}(x, \hat{\beta})) \mid \text{data} \xrightarrow{d} \Psi_{ij}(x, \beta_0) + \eta_{ij}(x, \beta_0)^\top \Sigma^{-1}(\tau, \beta_0) U(\tau, \beta_0)
\]
in probability, where \( U(x, \beta_0) \) and \( \Psi(x, \beta_0) \), \( x \in [0, \tau], \, \tau < \tau_0 \), are Gaussian processes as defined in Theorem~\ref{thm:weak_conv}.

\end{proposition}
\begin{proof}
    See \citet{dabrowska1995estimation}.
\end{proof}

\section{Simulation Studies}

In this section, we present simulation studies based on Section~\ref{sec:est} and \ref{sec:inference} to evaluate the performance of the semi-Markov model and the DSH estimator.

\subsection{Simulation Design}
The simulation studies were conducted using the R functions defined in our custom simulation framework. We assume a five-state model (Figure~\ref{fig:AD}) to mimic the progression of Alzheimer's disease \citep{brookmeyer2019multistate}, namely, CN (cognitively normal), MCI (mild cognitive impairment), SCI (severe cognitive impairment), AD (Alzheimer's disease) and Death. 

\begin{figure}[ht]
    \centering
    \includegraphics[width=0.95\linewidth]{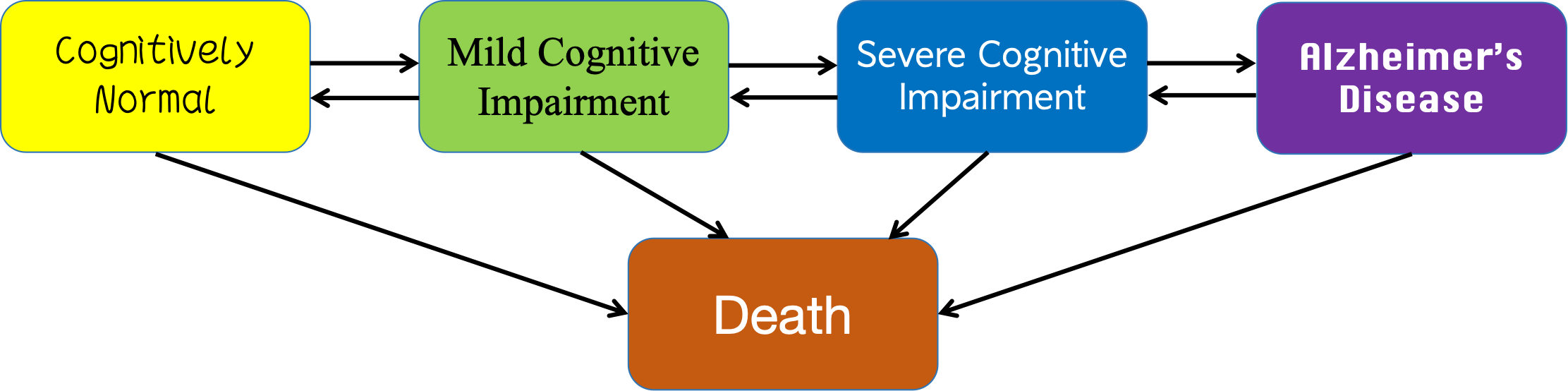}
    \caption{State transition diagram illustrating the possible transitions in the progression of Alzheimer's disease. Transient states include CN (cognitively normal), MCI (mild cognitive impairment), SCI (severe cognitive impairment), and AD (Alzheimer's disease), while absorbing state is death.}
    \label{fig:AD}
\end{figure}

Mimicking the parametric models commonly used \citep{aralis2016modeling}, we assume the baseline cumulative transition intensity matrix is of form
\begin{align*}
    \mathbf{A}_0(x)&=\left( \begin{matrix}
        0 & 0.2& 0& 0& 0.01\\
0.01& 0& 0.15& 0& 0.01&\\
 0& 0.01& 0& 0.15& 0.01\\
0& 0& 0.01& 0& 0.03\\
 0& 0& 0& 0& 0
    \end{matrix} \right)\times \frac{x^2}{2}.
\end{align*}
In other words, the baseline transition rate is of form $\alpha_{0ij}(x)=c_{ij}x$ where $c_{ij}$ is the $(ij)^{\text{th}}$ element of the above constant matrix. We note that this is not a homoegenuous Markov process since the transition rate now depends on the sojourn time $x$.   For each subject, the elements of transition-specific covariate matrix $\mathbf{Z}=(Z_{ij})$ are generated according to independent standard normal distributions. The regression coefficient $\beta$ is generated from $Uniform(-1,1)$ so that for transition $i\rightarrow j$, the cumulative transition rate is of form
$$A_{ij}(x)=A_{0ij}(x)e^{\beta^TZ_{ij}}.$$
The sojourn time $W_n$ then has survival function $S(W_n> x|X_n=i)=\exp\left(-\sum_{j}A_{ij}(x)\right)$ and by a simple distribution transformation trick, we can simulate $W_n$ from a uniform distribution. Next, given $W_n=x$ and $X_n=i$, we sample $X_{n+1}$ according to a multinomial distribution 
$$X_{n+1}\sim \mathcal{M}\left(1,\mathbf{p}\right),\ \mathbf{p}=\left(\frac{\alpha_{i1}(x)}{\sum_{j\not=i}\alpha_{ij}(x)},\cdots,\frac{\alpha_{i5}(x)}{\sum_{j\not=i}\alpha_{ij}(x)}\right)^T.$$
If $X_n=5$, i.e., death, then we stop the process; if $\sum_{n}W_n$ exceeds a pre-specified threshold, say $t$, then we stop the process by noting that the last sojourn time is right-censored. We summarize the data generating procedure in Algorithm~\ref{algs:simulation}.

\begin{algorithm}[H]
\caption{Data Generation Procedure for Semi-Markov Simulation}
\begin{algorithmic}[1]
\State Initialize \emph{M} subjects with starting state CN
\For{each subject $m = 1, ..., M$}
    \State Set current time $t = 0$, state $X_0 = 1$ (CN) and initial sojourn time $W_0=0$
    \State Generate covariate vector ${Z}_{mij} \sim N(0, I)$ for each pair of possible $(i,j)$
    \While{$t < T_{max}$ and $X_n \neq$ Death}
        \State Calculate cumulative transition intensity matrix $\mathbf{A}(x)$
        \State Generate sojourn time $W_n \sim S(W_n > x|X_n=i)$
        \State \hspace{1em} First generate $U\sim Uniform(0,1)$, then set $W_n=(\sum_jA_{ij})^{-1}(-\log(1-U))$
        \State Set censoring indicator $\Delta_n=1$
        \If{$t + W_n > T_{max}$}
            \State Right-censor the last sojourn time
            \State Set censoring indicator $\Delta_n=0$ and $X_n=X_{n-1}$
            \State Break
        \EndIf
        \State Draw next state $X_{n+1}$ from the  multinomial distribution $\mathcal{M}(1, \mathbf{p})$
        \State Update $t = t + W_n$
    \EndWhile
\EndFor
\State \Return $(X_n, W_n, \Delta_n)$ for $m=1, \cdots, M$
\end{algorithmic}\label{algs:simulation}
\end{algorithm}

\subsection{Simulation Results}
In our simulation studies, we set $T_{max}=50$ and $M=3000$ to ensure a sufficient large sample; the dimension of $Z_{mij}$ is set to $3$ so that the dimension of $\beta$ is 30 (\text{Dim}($Z_{ij}\times$ number of possible transitions)). For each simulated dataset, the multistate Cox model was fitted based on the \texttt{coxph} function for competing risks data. The function leverages the Breslow method for handling ties and computes transition probabilities using the proposed method (Section~\ref{sec:est}), which approximates higher-order convolution terms for enhanced predictive accuracy.

\begin{figure}[!ht]
    \centering
    \begin{minipage}[t]{0.45\textwidth}
        \centering
        \includegraphics[width=\linewidth]{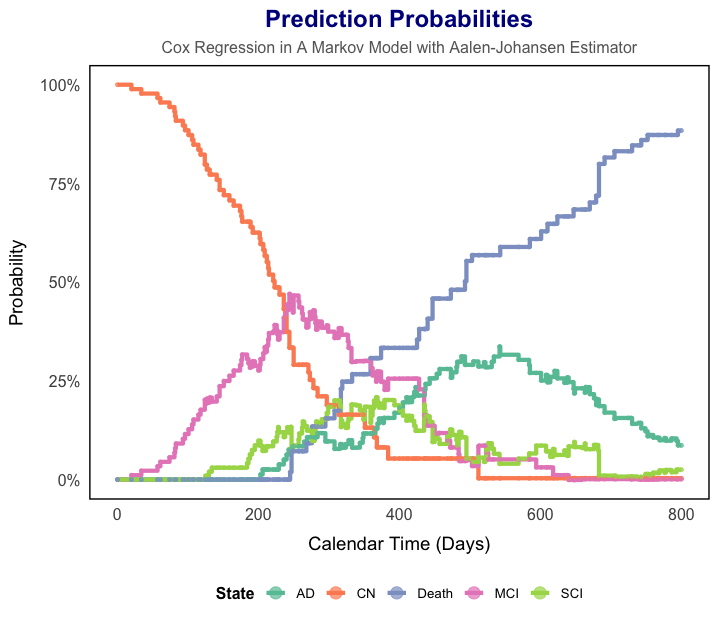}
        \caption{Prediction probabilities under a Markov model with the Aalen-Johansen estimator.}
        \label{fig:markov-aj-simu}
    \end{minipage}
    \hfill
    \begin{minipage}[t]{0.45\textwidth}
        \centering
        \includegraphics[width=\linewidth]{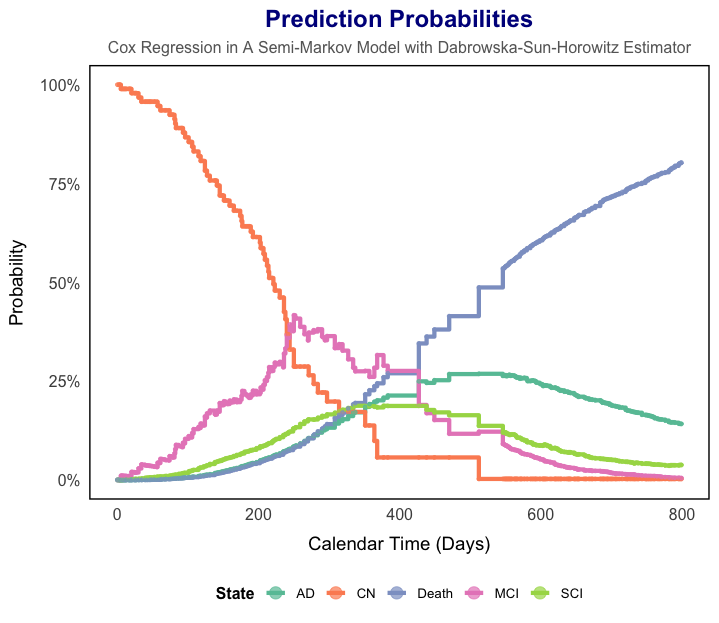}
        \caption{Prediction probabilities under a Semi-Markov model with the Dabrowska-Sun-Horowitz estimator.}
        \label{fig:semimarkov-dsh-simu}
    \end{minipage}
\end{figure}

Figures \ref{fig:markov-aj-simu} and \ref{fig:semimarkov-dsh-simu} present a comparison between the estimated transition probabilities using the AJ estimator under a Markov model and the DSH estimator under a semi-Markov model. Empirically, the DSH estimator appears smoother than the AJ estimator. While this observation lacks a formal theoretical justification, we conjecture that this smoothing effect may arise from the convolution process. Specifically, since the DSH estimator involves the renewal matrix, which is constructed as a summation of convolutions of different orders, the inherent property of convolutions to produce smoother functions might contribute to this observed effect.


\newpage
\section{The EBMT Example}
\subsection{Dataset}
We use the EBMT dataset from the \texttt{mstate} package \citep{de2011mstate, saccardi2023benchmarking} to illustrate the application of multi-state models in clinical research. This dataset includes data on 2,279 patients who underwent hematopoietic stem cell transplantation at the European Society for Blood and Marrow Transplantation (EBMT) between 1985 and 1998. The patients' clinical trajectories are captured across six key states: transplantation (TX), platelet recovery (PLT Recovery), adverse events (Adverse Event), concurrent recovery and adverse events (Recovered and Adverse Event), relapse while alive (Alive in Relapse), and relapse or death (Relapse/Death).

The figure (Figure~\ref{fig:transition-diagram}) illustrates the transitions between these states. TX represents the starting point for all patients, with possible transitions to PLT Recovery, Adverse Event, or directly to Relapse/Death. Similarly, patients in PLT Recovery or Adverse Event may transition to subsequent states such as Relapse/Death or Alive in Relapse. Alive in Relapse is a transient state, while Relapse/Death represents the absorbing state, indicating the end of the clinical trajectory.

\begin{table}[ht]
    \centering
    \caption{Summary of Transition Frequencies and Relative Frequencies}
    
    \label{tab:transition-data}
    \begin{tabular}{lllrr}
        \hline\hline
        \textbf{From State} & \textbf{To State} & \textbf{Frequency} & \textbf{Relative Frequency} \\ \hline
        TX                 & PLT              & 785               & 0.403 \\
        TX                 & AE               & 907               & 0.466 \\
        TX                 & Relapse          & 95                & 0.0488 \\
        TX                 & Death            & 160               & 0.0822 \\
        PLT                & Rec \& AE        & 227               & 0.601 \\
        PLT                & Relapse          & 112               & 0.296 \\
        PLT                & Death            & 39                & 0.103 \\
        AE                 & Rec \& AE        & 433               & 0.631 \\
        AE                 & Relapse          & 56                & 0.0816 \\
        AE                 & Death            & 197               & 0.287 \\
        Rec \& AE          & Relapse          & 107               & 0.439 \\
        Rec \& AE          & Death            & 137               & 0.561 \\ \hline\hline
    \end{tabular}
\end{table}

This dataset also includes several covariates, such as transplantation year, patient age, prophylaxis status, and donor-recipient gender matching, enabling covariate-adjusted analyses. Using a multi-state model, researchers can examine the dynamic interplay among recovery, adverse events, relapse, and death. These models provide insights into the factors influencing survival trajectories and patient outcomes following hematopoietic stem cell transplantation.

\begin{figure}[ht]
    \centering
    \includegraphics[width=0.5\linewidth]{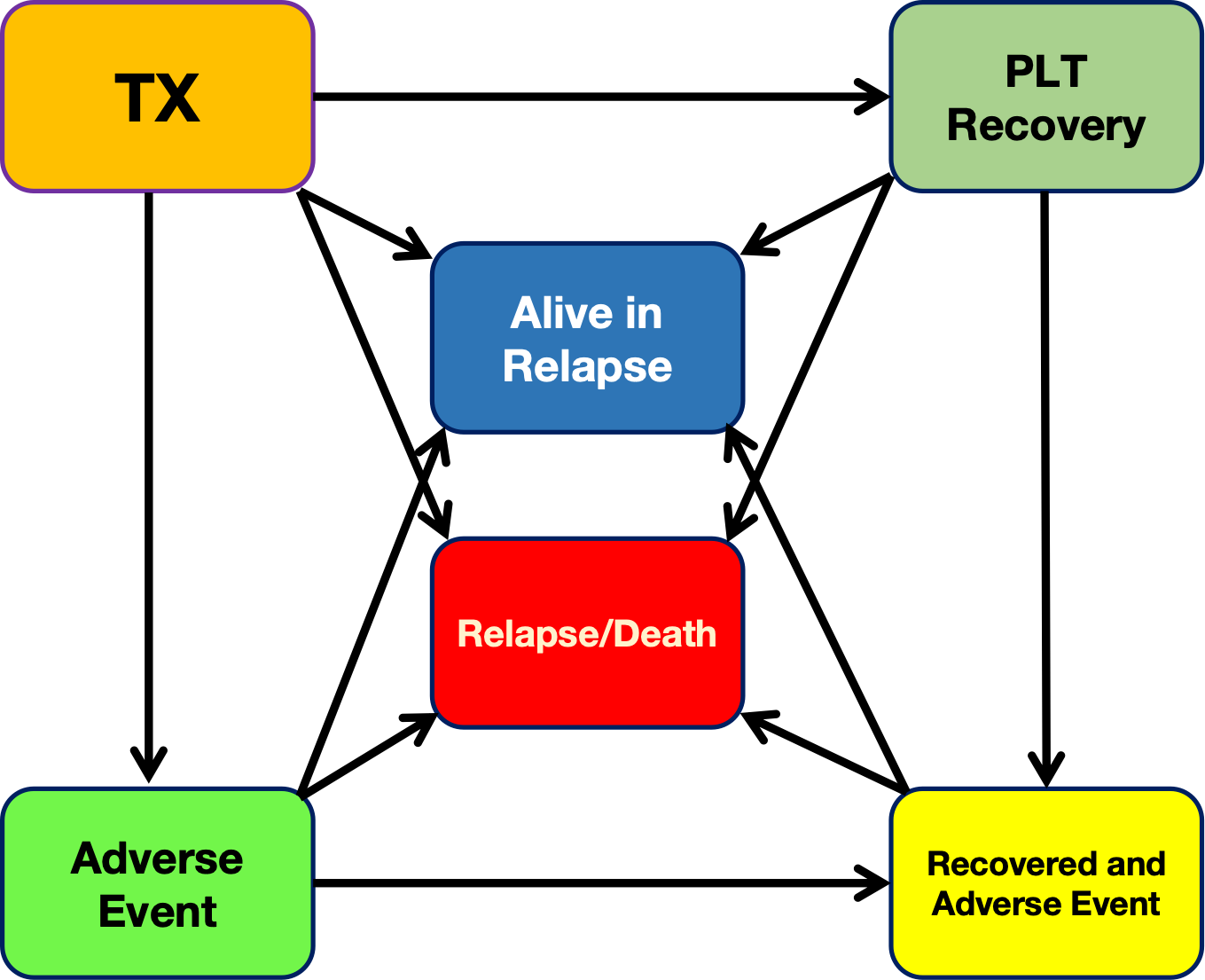}
    \caption{State transition diagram illustrating the possible transitions between clinical states in the EBMT dataset. Temporary states include TX (transplantation), PLT Recovery, Adverse Event, and Recovered and Adverse Event, while absorbing states are Alive in Relapse and Relapse/Death.}
    \label{fig:transition-diagram}
\end{figure}

\subsection{Pre-processing and Visualization}
The pre-processing steps involved preparing the EBMT dataset for multi-state analysis using the \texttt{mstate} package. First, the transition matrix was defined, specifying the allowed transitions between states, such as ``TX'' (transplantation) to ``PLT Recovery'' or ``Adverse Event.'' The dataset was then restructured using the \texttt{msprep} function, which maps the time-to-event data and corresponding event statuses to a multi-state format while retaining key covariates such as age group, prophylaxis status, donor-recipient match, and transplantation year. Dummy variables were created for the categorical variable \texttt{agecl} (age classes) to facilitate modeling, and these were appended to the dataset after removing the original \texttt{agecl} column. 

Additionally, the sojourn time (time spent in a state) was calculated as the difference between \texttt{Tstop} and \texttt{Tstart}. To standardize the time intervals, \texttt{Tstart} was reset to 0, and \texttt{Tstop} was updated to match the sojourn time. Furthermore, categorical variables such as \texttt{proph} (prophylaxis) and \texttt{match} (donor-recipient match) were transformed into binary numerical variables for analysis.  The final dataset, \texttt{ms\_data\_exp\_sojourn}, contains the cleaned and expanded covariates, ready for modeling and analysis. 

Figure~\ref{fig:sankey-diagram} provides an enhanced Sankey diagram illustrating the transitions between clinical states in the EBMT dataset. The diagram captures the flow of patients through six key states: ``TX'' (transplantation), ``PLT Recovery,'' ``Adverse Event,'' ``Recovered and Adverse Event,'' ``Alive in Relapse,'' and ``Relapse/Death.'' Each state represents a clinically significant phase in a patient's journey post-transplantation. Temporary states, such as TX, PLT Recovery, Adverse Event, and Recovered and Adverse Event, allow transitions to subsequent states, while Alive in Relapse and Relapse/Death are absorbing states, indicating the final clinical outcomes.

The width of each flow between states is proportional to the frequency of transitions, as summarized in Table~\ref{tab:transition-data}. For example, 46.6\% of patients transitioned from TX to Adverse Event, while 40.3\% transitioned to PLT Recovery. Only 4.88\% and 8.22\% of patients progressed directly from TX to Relapse and Death, respectively. Similarly, patients in PLT Recovery were most likely to transition to Recovered and Adverse Event (60.1\%), while 29.6\% transitioned to Relapse and 10.3\% to Death. The table also highlights the high likelihood of transition from Adverse Event to Recovered and Adverse Event (63.1\%) or directly to Death (28.7\%).

The visualization uses a color palette to distinguish between origin states, with flows smoothly curving to their respective destination states. Stratum boxes at each axis represent the states, with labels indicating the state names. The x-axis represents the sequence of transitions, while the y-axis represents the flow frequency. The Sankey diagram provides a comprehensive view of patient outcomes and the interplay between recovery, adverse events, relapse, and death.

\begin{figure}[ht]
    \centering
    \includegraphics[width=0.8\linewidth]{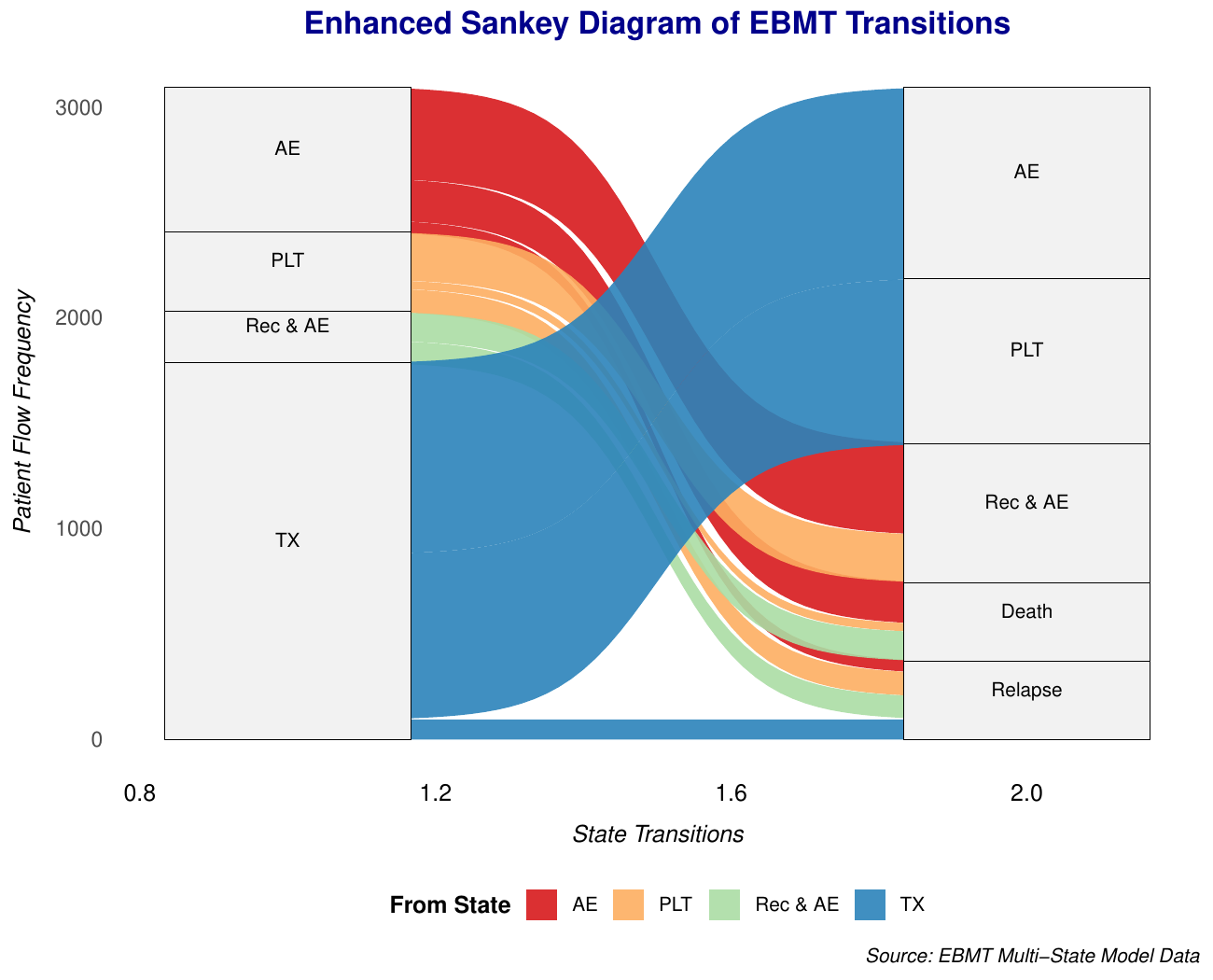}
    \caption{Enhanced Sankey diagram illustrating patient transitions across clinical states in the EBMT dataset. The flows represent transition frequencies between states, with temporary and absorbing states clearly indicated. The diagram provides insights into the dynamics of patient trajectories post-transplantation.}
    \label{fig:sankey-diagram}
\end{figure}

\subsection{Data Analysis}

To analyze the dynamic transition probabilities across clinical states in the EBMT dataset, we employed four different modeling approaches: 
\begin{enumerate}
    \item A Markov model with the Aalen-Johansen (AJ) estimator (Figure~\ref{fig:markov-aj}).
    \item A Markov model with the Dabrowska-Sun-Horowitz (DSH) estimator (Figure~\ref{fig:markov-dsh}).
    \item A Semi-Markov model with the Aalen-Johansen (AJ) estimator (Figure~\ref{fig:semi-markov-aj}).
    \item A Semi-Markov model with the Dabrowska-Sun-Horowitz (DSH) estimator (Figure~\ref{fig:semi-markov-dsh}).
\end{enumerate}
The estimated coefficients and standard deviations of the semi-Markov model with the DSH estimator is given in Table~\ref{tab:data_results}.

\begin{table}[ht]
\centering
\caption{Semi-Markov model with DSH estimator}
\begin{tabular}{lccccc}
\toprule
\textbf{Variable} & \textbf{Coef} & \textbf{exp(Coef)} & \textbf{SE(Coef)} & \textbf{z} & \textbf{p-value}  \\
\midrule
\texttt{agecl20-40}     & 0.108 & 1.114  & 0.0445 & 2.437 & 0.0148  \\
\texttt{agecl>40}       & 0.244 & 1.277  & 0.0515 & 4.758 & 1.96e-06  \\
\texttt{proph}          & -0.177 & 0.837  & 0.0422 & -4.212 & 2.53e-05  \\
\texttt{gender\_mismatch} & -0.0181 & 0.9821  & 0.0410 & -0.441 & 0.6591   \\
\bottomrule
\end{tabular}\label{tab:data_results}
\end{table}

\subsubsection{Key Observations}

\begin{enumerate}
    \item \textbf{Markov Model (AJ vs. DSH):}  
    In the Markov model, the Aalen-Johansen estimator (Figure~\ref{fig:markov-aj}) and the Dabrowska-Sun-Horowitz estimator (Figure~\ref{fig:markov-dsh}) showed similar trends in transition probabilities across states. However, the DSH estimator provided smoother trajectories, particularly for states such as ``Rec \& AE'' and ``Relapse.'' This indicates that the DSH estimator may better accommodate variations in transition intensities, particularly for less frequent transitions. Both models captured the rapid rise in the probability of ``PLT Recovery'' and the eventual dominance of absorbing states like ``Relapse/Death.''

    \item \textbf{Semi-Markov Model (AJ vs. DSH):}  
    The Semi-Markov models with the Aalen-Johansen (Figure~\ref{fig:semi-markov-aj}) and Dabrowska-Sun-Horowitz (Figure~\ref{fig:semi-markov-dsh}) estimators exhibited distinct differences. The AJ estimator tended to produce sharper changes in probabilities over time, particularly during transitions from temporary states (e.g., TX to PLT Recovery or Adverse Event). Conversely, the DSH estimator yielded smoother and more gradual transitions, capturing the inherent non-Markovian nature of the process more effectively. This smoothness was especially apparent in transitions to ``Rec \& AE'' and ``Relapse.''

    \item \textbf{Markov vs. Semi-Markov Models:}  
    The Markov models assumed memoryless transitions, leading to faster probabilities accumulating in absorbing states like ``Relapse/Death.'' In contrast, the Semi-Markov models accounted for sojourn times, resulting in delayed transitions to absorbing states. For example, the probability of ``Relapse/Death'' increased more gradually in the Semi-Markov models, particularly when using the DSH estimator. This highlights the importance of incorporating sojourn times to more accurately capture patient trajectories.

    \item \textbf{Effect of Estimators (AJ vs. DSH):}  
    Across both Markov and Semi-Markov frameworks, the DSH estimator consistently provided smoother and more interpretable transition probability curves compared to the AJ estimator. This suggests that the DSH estimator may be more suitable for clinical datasets where transitions are influenced by covariates and exhibit non-Markovian behavior.
\end{enumerate}

\begin{figure}[!ht]
    \centering
    \begin{minipage}[t]{0.45\textwidth}
        \centering
        \includegraphics[width=\linewidth]{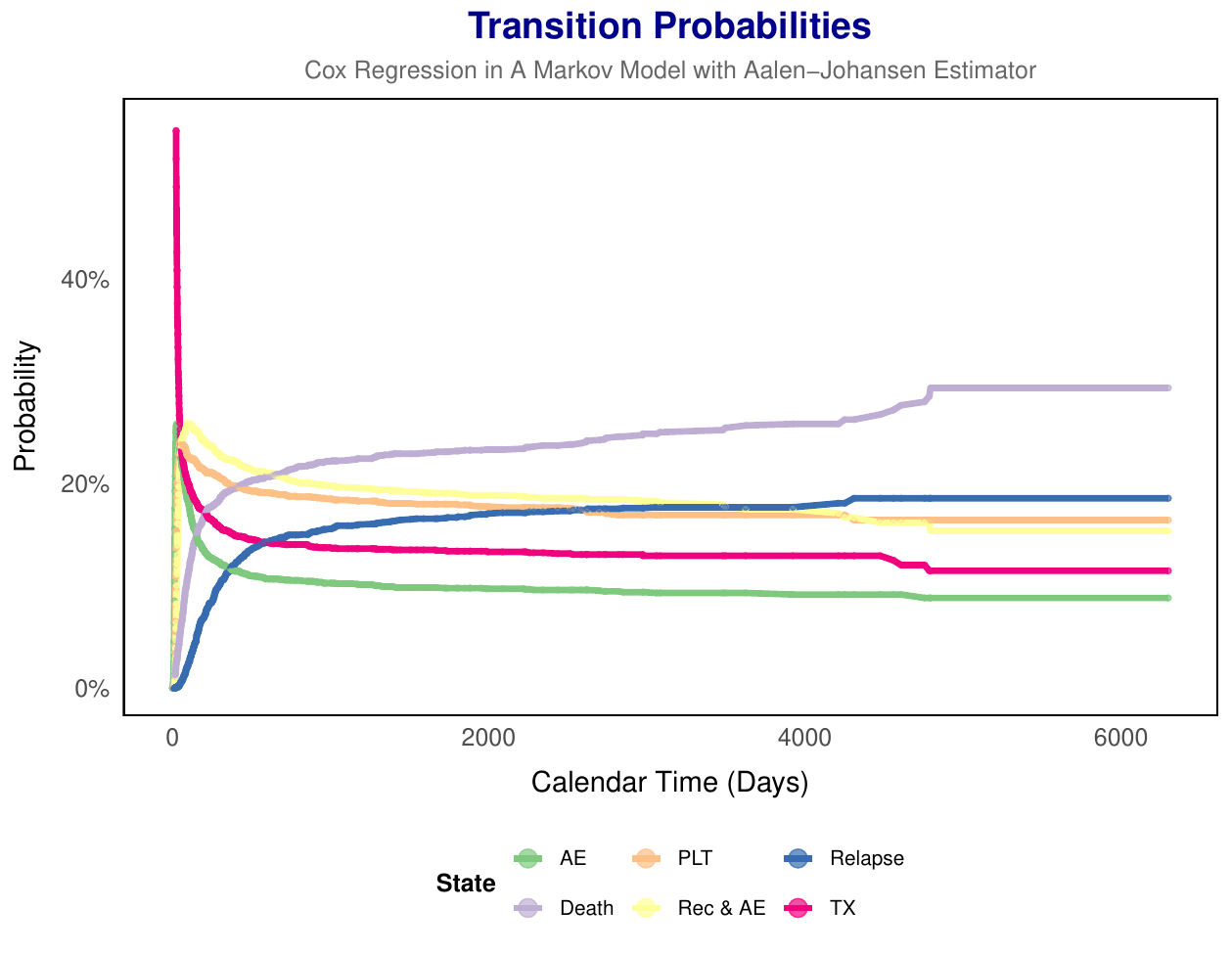}
        \caption{Transition probabilities under a Markov model with the Aalen-Johansen estimator.}
        \label{fig:markov-aj}
    \end{minipage}
    \hfill
    \begin{minipage}[t]{0.45\textwidth}
        \centering
        \includegraphics[width=\linewidth]{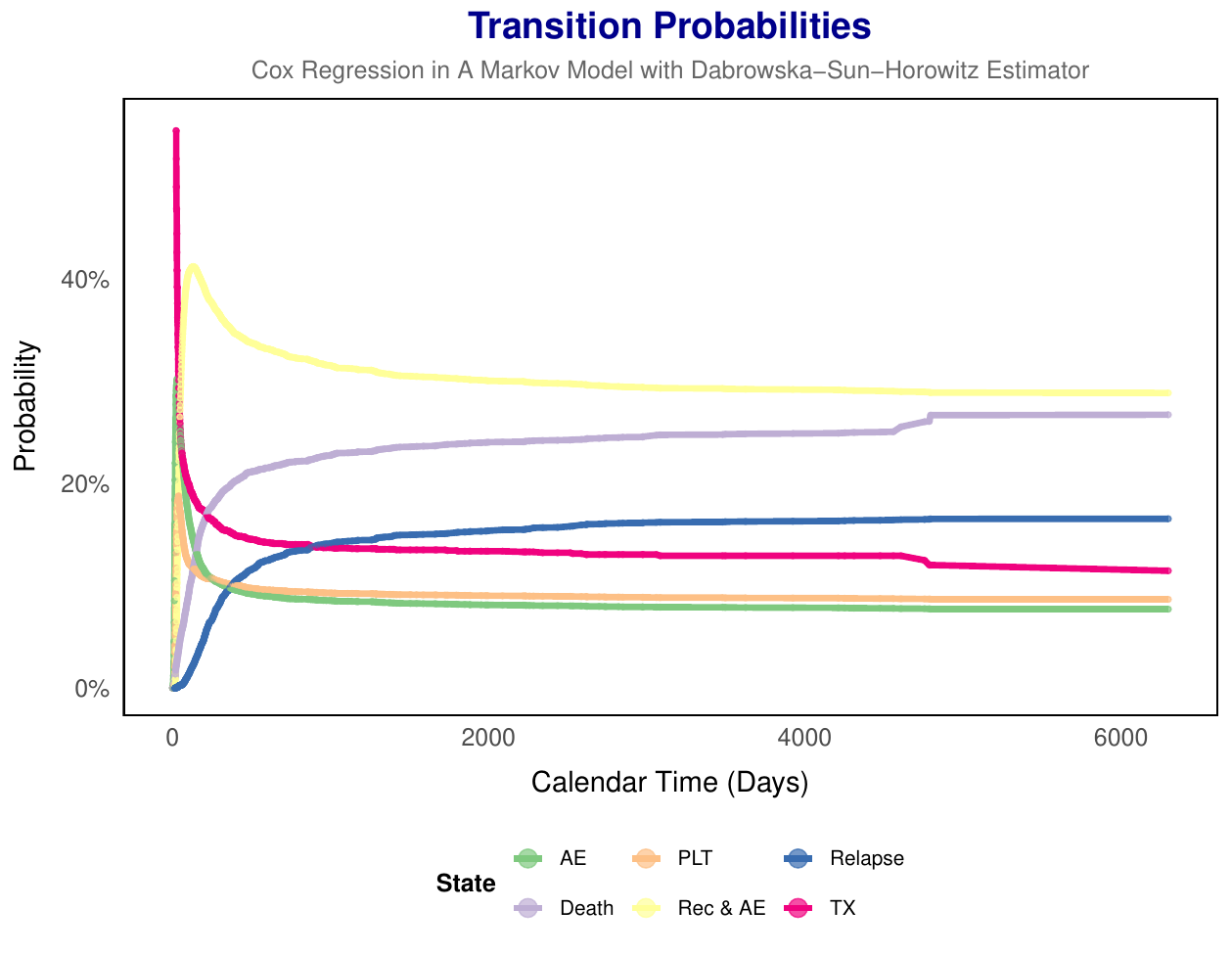}
        \caption{Transition probabilities under a Markov model with the Dabrowska-Sun-Horowitz estimator.}
        \label{fig:markov-dsh}
    \end{minipage}
    \vspace{0.5cm}  
    \begin{minipage}[t]{0.45\textwidth}
        \centering
        \includegraphics[width=\linewidth]{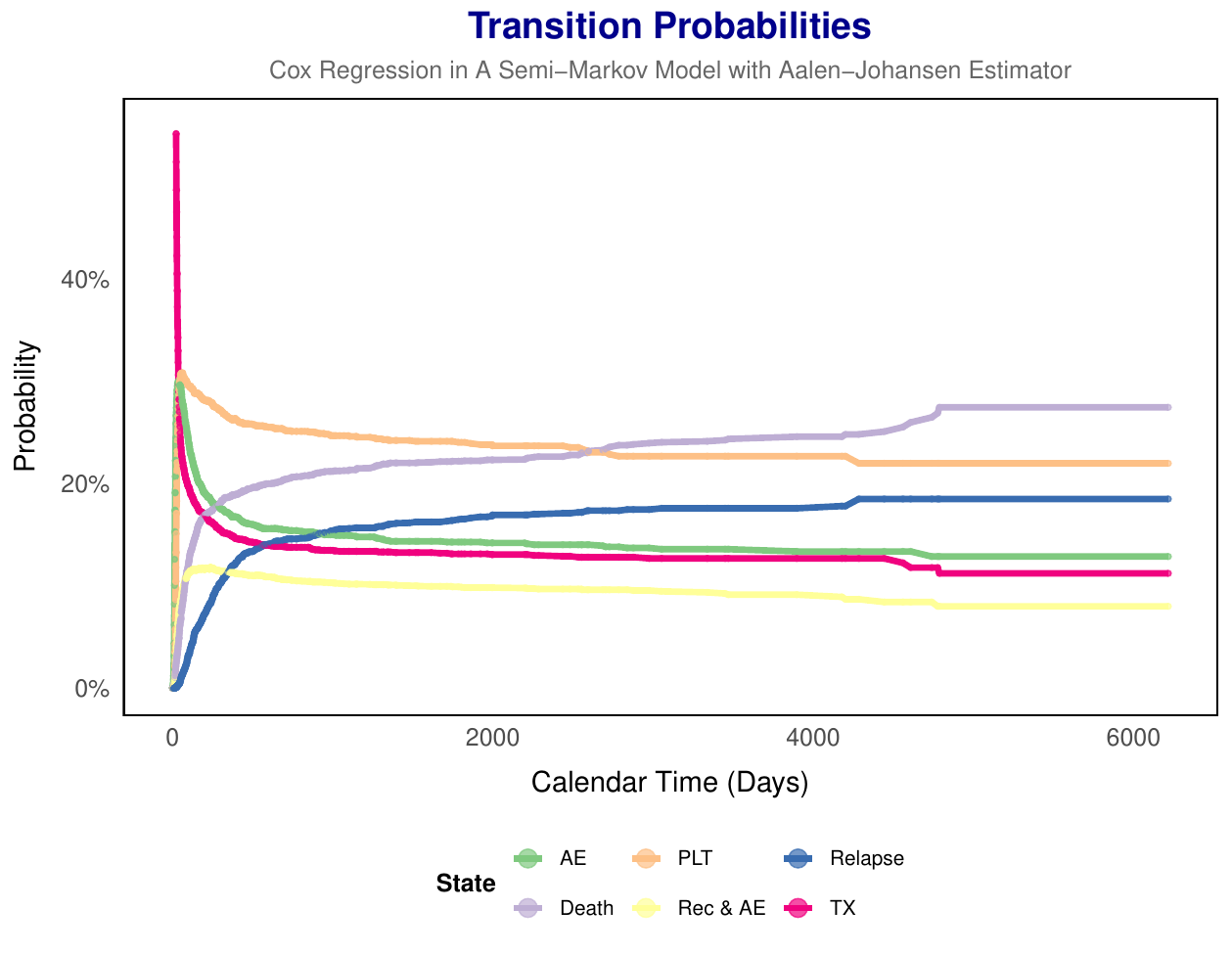}
        \caption{Transition probabilities under a Semi-Markov model with the Aalen-Johansen estimator.}
        \label{fig:semi-markov-aj}
    \end{minipage}
    \hfill
    \begin{minipage}[t]{0.45\textwidth}
        \centering
        \includegraphics[width=\linewidth]{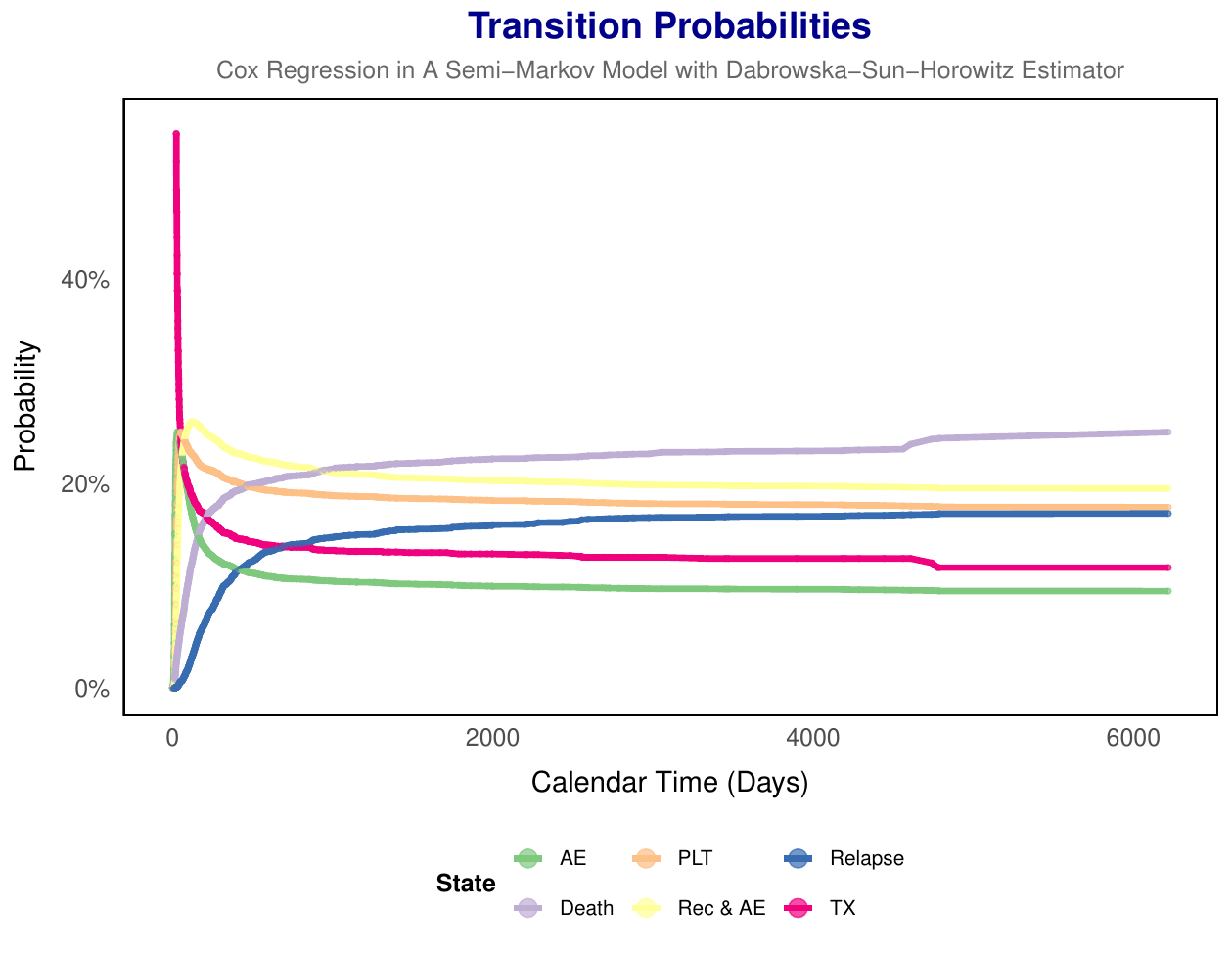}
        \caption{Transition probabilities under a Semi-Markov model with the Dabrowska-Sun-Horowitz estimator.}
        \label{fig:semi-markov-dsh}
    \end{minipage}
\end{figure}

\subsubsection{Clinical Implications}

These findings underscore the importance of selecting an appropriate modeling framework and estimator when analyzing multi-state processes. While the Markov models offer simplicity, the Semi-Markov models provide a more realistic representation of clinical pathways by incorporating sojourn times. Furthermore, the Dabrowska-Sun-Horowitz estimator's ability to produce smoother probability curves makes it a valuable tool for studying patient trajectories and understanding the nuanced effects of covariates.

\section{Conclusion and Discussion}

This paper investigates transition probability estimation under Markov and Semi-Markov frameworks, utilizing the Aalen-Johansen (AJ) and Dabrowska-Sun-Horowitz (DSH) estimators to analyze patient transitions in the EBMT dataset. The results highlight key distinctions between these modeling approaches and their implications for clinical research. The Markov model, which assumes memoryless transitions, demonstrated faster accumulation of probabilities in absorbing states such as ``Relapse'' and ``Death,'' as it does not account for the time spent in transient states. In contrast, the Semi-Markov model incorporates sojourn times, providing a more realistic and temporally sensitive depiction of patient trajectories. This adjustment delayed transitions to absorbing states, offering insights into the prolonged nature of certain clinical events. The choice of estimator further influenced the results. The AJ estimator captured sharper probability changes, reflecting its responsiveness to immediate variations in transition intensities. However, the DSH estimator produced smoother and more interpretable probability curves, effectively accommodating the complexities of non-Markovian behavior. These differences were particularly pronounced in states such as ``Recovered and Adverse Event'' or ``Relapse,'' where prolonged transitions occur.

These findings emphasize the importance of selecting appropriate models and estimators for analyzing multi-state processes. While Markov models are computationally efficient, they may oversimplify patient trajectories by neglecting sojourn times. Semi-Markov models, though computationally more demanding, offer a more accurate representation of real-world clinical scenarios, particularly when combined with the DSH estimator, which better reflects the nuances of patient transitions. Future research should focus on integrating advanced covariate structures, such as time-dependent and interaction effects, to further enhance model accuracy. Additionally, extending the frameworks to accommodate generalized non-Markovian models, such as phase-type distributions for sojourn times, may provide deeper insights. Improved visualization tools, like the enhanced Sankey diagram presented here, can also play a crucial role in communicating complex transition dynamics to a broader audience. By aligning statistical methodologies with the needs of clinical decision-making, these advancements have the potential to improve our understanding of patient trajectories, ultimately contributing to better clinical outcomes.

\section{Appendix}
\subsection{A Table of Notations}

\begin{table}[h!]
\centering
\caption{Summary of Notations}
\begin{tabular}{ll}
\toprule
\textbf{Symbol} & \textbf{Description} \\
\midrule
$X_n$ & State of the process after the $n$-th transition \\
$T_n$ & Time of the $n$-th transition \\
$W_n$ & Sojourn time in state $X_n$, $W_n = T_n - T_{n-1}$ \\
$Q_{ij}(x)$ & Semi-Markov kernel: Probability of transitioning from state $i$ to $j$ within sojourn time $x$ \\
$P_{ij}(s, t)$ & Transition probability matrix from state $i$ to $j$ between calendar times $s$ and $t$ \\
$A_{ij}(x)$ & Cumulative intensity matrix for transitions ($t$ for calendar time and $x$ for sojourn time) \\
$G_i(x)$ & Survival probability in state $i$, $G_i(x) = 1 - \sum_k Q_{ik}(x)$ ($t$ for calendar time and $x$ for sojourn time) \\
$N_{mij}$ & Counting process that counts transitions from $i$ to $j$ for $m$-th subject\\
\bottomrule
\end{tabular}
\label{tab:notations}
\end{table}

\subsection{Some Technical Difficulties}
As pointed out by Professor Dabrowska (personal communication), there is a drawback of semi-Markov Cox regression models that limits its usage to wide applications. We illustrate her point below.

Convolutions are computed as products of the transition probability matrix evaluated at appropriate points, followed by summing the resulting matrices. Importantly, after three or four convolutions, there is no change in the estimate.

It is therefore impractical and unnecessary to allocate resources to computations for the infinite convolution of the transition probability matrix. While theoretically, the error is still to be estimated. Mathematically, we have
$$\widehat{\mathbf{Q}}^{(p)}(x)=\int_0^x\widehat{\mathbf{Q}}(du)\widehat{\mathbf{Q}}^{(p-1)}(x-u)$$
and since $\widehat{\mathbf{Q}}$ is a jump process, we need both $\widehat{\mathbf{Q}}(\Delta u)$ and $\widehat{\mathbf{Q}}^{(p-1)}(\Delta (x-u))$ to be non-zero.

\subsection{Some Illustrative Examples}
\subsubsection{The Markov case}
In the Markov case, we need to compute the product integral of form $\Prodi_{u\in(s,t]}(\mathbf{I}+\widehat{\mathbf{A}}(du))$, which only involves finite terms (the number of distinct event times plus $1$). For example, if the $i$-th event time is $t_i$, then $\mathbf{P}(0, t)$ for $t\in(t_2,t_3]$ can be estimated as
$$\widehat{\mathbf{P}}(0,t)=(\mathbf{I}+\widehat{\mathbf{A}}(\Delta t_1))(\mathbf{I}+\widehat{\mathbf{A}}(\Delta t_2)).$$
If we assume the first transition is $1\rightarrow 2$ and the second is $1\rightarrow 3$, then in the nonparametric case we have
\begin{align*}
    \widehat{\mathbf{A}}(\Delta t_1))&=\left(\begin{matrix}
        -\frac{N_{12}(\Delta t_1)}{Y_i(t_1)}& \frac{N_{12}(\Delta t_1)}{Y_i(t_1)} & 0&\cdots &0\\
        0&0&0&\cdots& 0\\
        \vdots&\vdots&\vdots&\vdots&\vdots\\
        0&\cdots &\cdots &\cdots &0
    \end{matrix}\right),   \\  
    \widehat{\mathbf{A}}(\Delta t_2))&=\left(\begin{matrix}
        -\frac{N_{13}(\Delta t_2)}{Y_i(t_2)}& 0& \frac{N_{13}(\Delta t_1)}{Y_i(t_2)} &\cdots &0\\
        0&0&0&\cdots& 0\\
        \vdots&\vdots&\vdots&\vdots&\vdots\\
        0&\cdots &\cdots &\cdots &0
    \end{matrix}\right)
\end{align*}
and
\[
\widehat{\mathbf{P}}(0,t) =
\begin{bmatrix}
\left(1-\frac{N_{12}(\Delta t_1)}{Y_i(t_1)}\right)\left(1-\frac{N_{13}(\Delta t_2)}{Y_i(t_2)}\right) & \frac{N_{12}(\Delta t_1)}{Y_i(t_1)} & \frac{N_{13}(\Delta t_2)}{Y_i(t_2)}\left(1-\frac{N_{12}(\Delta t_1)}{Y_i(t_1)}\right) &\cdots & 0 \\
0 & 1 & 0 & \cdots & 0 \\
\vdots & \vdots & \vdots & \vdots & \vdots \\
0 & \cdots & \cdots & 0 & 1
\end{bmatrix}
\]
for any $t\in(t_2,t_3]$. The cumulative hazard $A_{ij}$ and its increments can be computed from standard packages in $R$. In the semi-parametric case, we replace $\frac{N_{ij}(\Delta t_k)}{Y_i(t_k)}$ with $\frac{N_{ij}(\Delta t_k)}{mS_{ij}^{(0)}(t_k,\beta)}$. For prediction probabilities, we have
$$\widehat{\mathbf{P}}(s,t)=\frac{\widehat{\mathbf{P}}(0,t)}{\widehat{\mathbf{P}}(0,s)}=\Prodi_{u\in(s,t]}(\mathbf{I}+\widehat{\mathbf{A}}(du)).$$

\subsubsection{The semi-Markov case}
In the semi-Markov case, we first need to compute the one-step transition probability or the semi-Markov kernel $\widehat{\mathbf{Q}}(x)=\int_0^x\widehat{\mathbf{G}}(u-)\widehat{\mathbf{A}}(du)$. If we assume the first transition is $1\rightarrow 2$ and the second is $1\rightarrow 3$ (in terms of sojourn time), then in the nonparametric case we have
\begin{align*}
    \widehat{\mathbf{A}}(\Delta x_1)&=\left(\begin{matrix}
        0& \frac{N_{12}(\Delta x_1)}{Y_i(x_1)} & 0&\cdots &0\\
        0&0&0&\cdots& 0\\
        \vdots&\vdots&\vdots&\vdots&\vdots\\
        0&\cdots &\cdots &\cdots &0
    \end{matrix}\right),\ \widehat{\mathbf{G}}(x_1-)=\left(\begin{matrix}
        1& 0 & 0&\cdots &0\\
        0&1&0&\cdots& 0\\
        \vdots&\vdots&\vdots&\vdots&\vdots\\
        0&\cdots &\cdots &\cdots &1
    \end{matrix}\right),   \\  
    \widehat{\mathbf{A}}(\Delta x_2)&=\left(\begin{matrix}
        0& 0& \frac{N_{13}(\Delta x_1)}{Y_i(x_2)} &\cdots &0\\
        0&0&0&\cdots& 0\\
        \vdots&\vdots&\vdots&\vdots&\vdots\\
        0&\cdots &\cdots &\cdots &0
    \end{matrix}\right),\ \widehat{\mathbf{G}}(x_2-)=\left(\begin{matrix}
        1-\widehat{A}_{12}(x_1)& 0 & 0&\cdots &0\\
        0&1&0&\cdots& 0\\
        \vdots&\vdots&\vdots&\vdots&\vdots\\
        0&\cdots &\cdots &\cdots &1
    \end{matrix}\right).
\end{align*}
Hence, for $t\in(x_2,x_3]$, we have
\begin{align*}
    \widehat{\mathbf{Q}}(t)&=\widehat{\mathbf{A}}(\Delta x_1)\widehat{\mathbf{G}}(x_1-))+ \widehat{\mathbf{A}}(\Delta x_2))\widehat{\mathbf{G}}(x_2-)\\
    &=\widehat{\mathbf{A}}(\Delta x_1)+\widehat{\mathbf{A}}(\Delta x_2).
\end{align*}
The two-step transition matrix is
\begin{align*}
     \widehat{\mathbf{Q}}^{(2)}(t)&=\int_0^t \widehat{\mathbf{Q}}(du)\widehat{\mathbf{Q}}(t-u)\\
     &=\widehat{\mathbf{Q}}(\Delta x_1)\widehat{\mathbf{Q}}(t-x_1) + \widehat{\mathbf{Q}}(\Delta x_2)\widehat{\mathbf{Q}}(t-x_2).
\end{align*}
The renewal matrix is estimated by $\widehat{\mathbf{R}}(t)=\mathbf{I} + \widehat{\mathbf{Q}}(t) + \widehat{\mathbf{Q}}^{(2)}(t)$ and finally we have the transition probability matrix estimated by
\begin{align*}
    \mathbf{\widehat{P}}(t)&=\int_0^t\widehat{\mathbf{R}}(du)\widehat{\mathbf{G}}(t-u)\\
    &=\widehat{\mathbf{R}}(\Delta x_1)\widehat{\mathbf{G}}(t-x_1)+\widehat{\mathbf{R}}(\Delta x_2)\widehat{\mathbf{G}}(t-x_2).
\end{align*}
The semi-parametric case can be derived similarly.

\subsection{Packages in R}
We provide a table that summarizes existing packages for multi-state models in $R$.

\begin{sidewaystable}[!htbp]
\caption{Comparison of R Packages for Survival and Multi-State Models.}
\renewcommand{\arraystretch}{1.2} 
\centering
\begin{tabular}{p{2.2cm}p{3cm}p{4cm}p{4cm}p{3cm}}
\hline\hline
\textbf{R Package} & \textbf{Targeted Dataset} & \textbf{Pros} & \textbf{Cons} & \textbf{References} \\ 
\hline
\texttt{mstate} & Multi-state and competing risks data & Comprehensive analysis of transition probabilities & Complex syntax for beginners & \citet{de2011mstate} \\ 

\texttt{msm} & Panel data in multi-state models & Flexible for general multi-state modeling & Limited support for time-dependent covariates & \citet{jackson2011multi} \\ 
\texttt{TPmsm} & Three-state illness-death models & Techniques for handling heavily censored cases & Limited to three-state models & \citet{araujo2015tpmsm} \\

\texttt{timereg} & Competing risks and survival data & Advanced support for regression in competing risks & Requires significant statistical expertise & \citet{scheike2011analyzing} \\ 

\texttt{msSurv} & Nonparametric estimation for multi-state models & Nonparametric estimation without parametric assumptions & Focused on nonparametric methods & \citet{zhao2012mssurv} \\ 

\texttt{SemiMarkov} & Semi-Markov models & Handles semi-Markov processes & Only deals with parametric case & \citet{kruger2015semimarkov} \\ 

\texttt{ctmcd} & Markov chains for discrete-time data & Simple for discrete-time Markov modeling & Only applicable to discrete-time processes & \citet{wilkinson2017ctmcd} \\ 

\texttt{hhsmm} & Hidden hybrid Markov/semi-Markov models & Supports hybrid models with Markov/semi-Markov features & High computational complexity & \citet{spezia2022hhsmm} \\ 

\texttt{survminer} & Survival analysis visualization & User-friendly for Kaplan-Meier and Cox models & Limited functionality for advanced modeling & \citet{kassambara2023survminer} \\ 

\texttt{survival} & General survival analysis & Widely used, extensive documentation & Limited built-in support for visualization & \citet{therneau2020survival} \\ 

\texttt{cmprisk} & Competing risks data & Widely used, extensive documentation & Not for general multi-state models & \citet{gray2014package} \\ 

\texttt{eha} & Event history analysis & Flexible handling of survival and transition models & Limited support for multi-state models & \citet{brostrom2024eha} \\ 

\texttt{mvna} & Multi-state models for nonparametric Aalen-Johansen estimators & Straightforward estimation of cumulative transition probabilities & Limited to nonparametric estimation & \citet{allignol2008mvna} \\ 

\texttt{flexsurv} & Parametric multi-state models & Interpretable and fast computation & Limited to parametric estimation & \citet{jackson2016flexsurv}\\
\hline\hline
\end{tabular}
\end{sidewaystable}

\newpage
\subsection{Proof of Theorem~\ref{thm:weak_conv}}\label{sec:proof_weak_conv}

\citet{andersen1982cox} utilized Taylor expansion, Rebolledo's martingale CLT, and Lenglart's inequality to establish consistency and asymptotic normality. However, in the Markov renewal setting, the latter two tools fail because the quantity 
\[ 
N_{ij}(x) - \int_0^x Y_i(u) \alpha_{ij}(du) 
\]
does not satisfy the martingale property. To address this issue, \citet{sun1992markov} and \citet{dabrowska1995estimation} employed Billingsley's theorem \citep{billingsley2013convergence} as an alternative to Rebolledo's martingale CLT and used the functional strong law of large numbers (SLLN) to replace Lenglart's inequality. Despite its robustness, their original proof lacks organization and remains challenging for new learners. In this work, we re-organize their approach to make it more accessible and reader-friendly.

The key idea begins with the Taylor expansion of $\widehat{U}(x, \beta)$ around the true parameter value $\beta_0$, expressed as:
\[
\widehat{U}(x, \beta) = \widehat{U}(x, \beta_0) - \mathscr{I}(x, \beta^*)(\beta - \beta_0),
\]
where 
\[
\mathscr{I}(x, \beta) = \sum_m \sum_{i \neq j} \int_0^x V_{ij}(u, \beta) N_{mij}(du),
\]
and $\beta^*$ lies on the line segment between $\beta$ and $\beta_0$. Plugging in $\beta = \widehat{\beta}$ and rearranging terms gives:
\[
\sqrt{m}(\widehat{\beta} - \beta_0) = \left(\frac{\mathscr{I}(\tau, \beta^*)}{m}\right)^{-1} \frac{1}{\sqrt{m}} \widehat{U}(\tau, \beta_0).
\]

The weak convergence of $\widehat{\beta}$ is then derived from three components: the consistency of $\widehat{\beta}$, the uniform convergence of $\mathscr{I}(x, \beta^*)$, and the weak convergence of $\widehat{U}(x, \beta_0)$. Similarly, the weak convergence of $\widehat{A}_{0ij}$ and the asymptotic independence follow from an analogous argument in Chapter VII.2 of \citet{andersen2012statistical}.

\subsubsection{Assumptions}\label{sec:regularity}
The following conditions are specified for $i,j\in\{1,2,\cdots,r\}$.
\begin{itemize}
    \item[(i)] $Z_{mij}$'s are bounded a.s.;
    \item[(ii)] There exist a neighborhood of $\beta_0$ (say $B$), such that for $k=0,1,2$,
    $$\sup_{\beta\in B,x\in[0,\tau]}\lVert\frac{1}{m}S_{ij}^{(k)}(x,\beta)-s_{ij}^{(k)}(x,\beta)\lVert\rightarrow 0$$
    in probability as $m\rightarrow\infty$ (this can be verified using functional LLN);
    \item[(iii)] $s_{ij}^{(0)}(\cdot,\beta_0)$ is bounded away from $0$ for $x\in[0,\tau]$;
    
    \item[(iv)] The matrix $\Sigma(x,\beta)=\sum_{i,j\le r}\int_0^xv_{ij}(u,\beta)s_{ij}^{(0)}(u,\beta)\alpha_{0,ij}(u)du$ is positive definite;
    \item[(v)] $A_{0ij}(\tau)=\int_0^\tau\alpha_{0ij}(x)dx<\infty$.
\end{itemize}

\subsubsection{Consistency of $\hat{\beta}$}
The key to prove consistency of $\widehat{\beta}$ is to show that the profile log-likelihood function $C(\tau,\beta)$ converges pointwise to a concave limit (in $\beta$), say, $f(\beta)$. Then by Theorem II.1 of \citet{andersen1982cox} we have that the maximizer of $C(\tau,\beta)$ converges in probability to the maximizer of $f(\beta)$, which is indeed $\beta_0$.
\begin{lemma}[\citet{sun1992markov}] $\widehat{\beta}$ is weakly consistent.
    
\end{lemma}
\begin{proof}
    We consider the limit of $\frac{1}{m}\left(C(\tau,\beta)-C(\tau,\beta_0)\right)$. As $m\rightarrow\infty$, we have
$$\frac{1}{m}\sum_mN_{mij}(x)\rightarrow\mathbb{E}\left(Y_{1i}(x)e^{\beta^T_0Z_{1ij}}\right)A_{0ij}(x)$$ by condition (iii) and the convergence is uniform in $x\in[0,\tau]$ by Glivenko-Cantelli (see also formula (27) in \citet{gill1980nonparametric}). By condition (i) and (ii), we have 
\begin{align*}
    &\frac{1}{m}\left(C(\tau,\beta)-C(\tau,\beta_0)\right)\\\rightarrow&\sum_{i\not=j}\int_0^\tau \left[(\beta-\beta_0)^Ts_{ij}^{(1)}(x,\beta_0)-\log\left(\frac{s_{ij}^{(0)}(x,\beta)}{s_{ij}^{(0)}(x,\beta_0)}\right)s_{ij}^{(0)}(x,\beta_0)\right]A_{0ij}(dx)
\end{align*}
in probability. Denote the limit as $f(\beta)$, we have by Dominated convergence
$$\frac{\partial}{\partial\beta}f(\beta)=\sum_{i\not=j}\int_0^\tau \left(e_{ij}(x,\beta_0)-e_{ij}(x,\beta)\right)s_{ij}^{(0)}(x,\beta_0)A_{0ij}(dx)$$
where $e_{ij}(x,\beta)=\frac{s^{(1)}_{ij}(x,\beta)}{s^{(0)}_{ij}(x,\beta)}$.
Clearly $\beta_0$ is one of the roots for the above to be 0. Next,
$$\frac{\partial^2}{\partial\beta\partial\beta^T}f(\beta)=-\sum_{i\not=j}\int_0^\tau v_{ij}(x,\beta)s^{(0)}_{ij}(x,\beta_0)A_{0ij}(dx)$$
and by Condition (iv), the above matrix is negative definite when $\beta=\beta_0$. This ensures that $\beta_0$ is the unique global maximizer of $f(\beta)$.  Applying Theorem II.1 in \citet{andersen1982cox}, the weak consistency of $\widehat{\beta}$ follows.
\end{proof}

\subsubsection{Weak Convergence of $\hat{U}$} We need to following representation lemma according to \citet{dabrowska1994cox}.

\begin{lemma}\label{lemma:representation} The processes $\widehat{U}$ and $\widehat{A}$ can be represented as
\begin{align*}
    \frac{1}{\sqrt{m}}\widehat{U}(x,\beta_0)&=\sum_{i\not=j}M_{ij}^{(1)}(x,\beta_0)-\sum_{i\not=j}\int_0^x E_{ij}(u,\beta_0)M_{ij}^{(0)}(du,\beta_0)\\
    \sqrt{m}\left(\widehat{A}_{0ij}(x,\widehat{\beta})-A_{0ij}(x,\beta_0)\right)&=\widehat{\Psi}_{ij}(x,\beta_0)+\widehat{\eta}_{ij}(x,\beta^*)^T\sqrt{m}(\widehat{\beta}-\beta_0)
\end{align*}
where $\beta^*$ lies on the line segment between $\beta$ and $\widehat{\beta}$,
\begin{align*}
    M_{ij}^{(k)}(x,\beta_0)&=\frac{1}{\sqrt{m}}\sum_m\int_0^x Z_{mij}^{\otimes k}\underbrace{\left(N_{mij}(du)-Y_{mi}(u)e^{\beta_0^TZ_{mij}}A_{0ij}(du)\right)}_{M_{mij}(du,\beta_0)}\\
    &=\frac{1}{\sqrt{m}}\sum_mZ_{mij}^{\otimes k}M_{mij}(x,\beta_0),\\
    \widehat{\Psi}_{ij}(x,\beta)&=\int_0^x\frac{\mathbb{I}(S_{ij}^{(0)}(u,\beta)>0)}{S_{ij}^{(0)}(u,\beta)}M_{ij}^{(0)}(du,\beta),\\
    \widehat{\eta}_{ij}(x,\beta)&=-\int_0^x E_{ij}(u,\beta)\widehat{A}_{0ij}(du,\beta_0).
\end{align*}
    
\end{lemma}
The proof follows from simple algebra and Taylor expansion of $\widehat{A}_{0ij}(x,\widehat{\beta})$ around $\beta_0$. The lemma entails us that the weak convergence of $\widehat{U}$ and $\widehat{A}$ can be derived from the weak convergence of $M_{ij}^{(0)}$ and $M_{ij}^{(1)}$, and this is the key for the whole theorem. Note that we write $M_{ij}^{(k)}$ the integral form because the covariate vector $Z_{mij}$ can be replaced by a predictable process $Z_{mij}(L(t))$ where $t$ is the calendar time and $L(t)$ is the backwards recurrence time.

\begin{lemma}[\citet{dabrowska1995estimation}] The process $\mathbf{M}(x,\beta_0)=\left( M_{ij}^{(k)}(x,\beta_0):k=0,1; i\not=j \right)$ converges weakly to a mean zero ($r(r-1)(1+d)$-dimensional) Gaussian process $$\mathbf{W}(x,\beta_0)=\left( W_{ij}^{(k)}(x,\beta_0):k=0,1;i\not=j \right).$$
The Gaussian process $\mathbf{W}$ has independent components, each with covariance
$$Cov\left( W_{ij}^{(k_1)}(x,\beta_0), W_{ij}^{(k_2)}(y,\beta_0) \right)=\int_0^{x\vee y}s^{(k_1+k_2)}_{ij}(u,\beta_0)A_{0ij}(du).$$
    
\end{lemma}
\begin{proof}
    The proof depends on  Billingsley's theorem on weak convergence of stochastic processes in space $D$ (See Theorem 13.5 in \citet{billingsley2013convergence} or Theorem 1.14.15 in \citet{vaart2023empirical}). Two ingredients ensures the weak convergence of $\mathbf{M}$: finite dimensional weak convergence of $\mathbf{M}$ and a mixed moment bound that only depends on three different time points of $\mathbf{M}$. We summarize the theorem below for reference.
    \begin{theorem}[\citet{billingsley2013convergence}] Let $$M_{ij}(x,\beta_0)=\frac{1}{\sqrt{m}}\sum_m\int_0^xg_{mij}(u)M_{mij}(du,\beta_0),i\not=j,$$
    where $g_{mij}$'s are bounded functions.
    Suppose that $(M_{ij}(x_1),\cdots,M_{ij}(x_p))$ converges weakly to some limiting process $\mathbf{W}$ for any finite $p$ and $x\in[0,\tau]$, and that for $0\le u<x<y\le \tau$,
    $$\mathbb{E}\left(|M_{ij}(x)-M_{ij}(u)|^2|M_{ij}(y)-M_{ij}(x)|^2\right)\le|F(y)-F(u)|^{1+\alpha}$$
    where $\alpha >0$ and $F$ is a continuous, nondecreasing function on $[0,\tau]$. Then $\mathbf{M}=[M_{ij}]\Rightarrow \mathbf{W}$ in space $D^{d}$ where $d$ is the dimension of $\mathbf{M}$.
        
    \end{theorem}
The finite dimensional convergence is simple. For $k=0$ or $1$, $M_{ij}^{(k)}(x,\beta_0)$ is just a normalized sum of i.i.d. zero mean random variables and its limiting distribution is given by the classical CLT,
$$M_{ij}^{(k)}(x,\beta_0)\Rightarrow \mathcal{N}\left(0,\ \mathbb{E} \left(M_{mij}^{(k)}(x,\beta_0)\right)^2\right).$$
The second moment is computed explicitly by formula (28) in \citet{gill1980nonparametric}, which is
$$\int_0^x\mathbb{E}\left(Y_{1i}(u)Z_{1ij}^{\otimes k}e^{\beta_0^TZ_{1ij}}\right)A_{0ij}(du)=\int_0^xs^{(k)}_{ij}(u,\beta_0)\alpha_{0ij}(u)du.$$
\citet{gill1980nonparametric} also computes that the covariance between components of $\mathbf{M}=[M_{ij}]$ is zero. So by Cramér-Wold device, the finite dimensional convergence follows and the limiting distribution is multivariate Gaussian with independent components.

For the mixed moment bound, we denote $M_{mij}(u,x)=M_{mij}(x)-M_{mij}(u)$ for short. Then we have
\begin{align*}
    &\mathbb{E}\left(|M_{ij}(x)-M_{ij}(u)|^2|M_{ij}(y)-M_{ij}(x)|^2\right)\\=&\frac{1}{m^2}\sum_a\sum_b\sum_c\sum_d \mathbb{E}\left(M_{aij}(u,x)M_{bij}(u,x)M_{cij}(x,y)M_{dij}(x,y)\right).
\end{align*}
Clearly, the summation is not zero only when $a=b=c=d$ or there are exactly two unequal pairs, i.e. $a=b\not=c=d$ or $a=c\not=b=d$ or $a=d\not=b=c$. Hence, the summation reduces to
\begin{align*}
    &\frac{1}{m^2}\sum_a\mathbb{E}M_{aij}^2(u,x)M_{aij}^2(x,y)\\+&\frac{1}{m^2}\sum_a\sum_{c\not=a}\mathbb{E}M_{aij}^2(u,x)\mathbb{E}M_{cij}^2(x,y)\\+&\frac{2}{m^2}\sum_a\sum_{b\not=a}\mathbb{E}M_{aij}(u,x)M_{aij}(x,y)\mathbb{E}M_{bij}(u,x)M_{bij}(x,y).
\end{align*}
    Although $M_{aij}(x)$ is not an ``actual" martingale, it behaves like a martingale in the sense that $M_{aij}(u,x)$ has uncorrelated increments and zero mean (again, this is due to formula (28) in \citet{gill1980nonparametric}, see also Chapter X.1 in \citet{andersen2012statistical}). Hence, the third term is $0$. For the first term, we use Cauchy-Schwarz,
    \begin{align*}
        \mathbb{E}M_{aij}^2(u,x)M_{aij}^2(x,y)&\le\sqrt{\mathbb{E}M_{aij}^4(u,x)\mathbb{E}M_{aij}^4(x,y)}.
    \end{align*}
    To handle the fourth moment, we first note that $\mathbb{E}M_{aij}^4(\tau)$ is the same as the expectation of the fourth moment of a counting process martingale associated with the Markov renewal process evaluated at $\tau$. Indeed,
    $$\int_0^\infty \mathbb{I}(u\le x)M^4_{aij}(du)=\int_0^\infty \mathbb{I}(L(t)\le x)\widetilde{M}^4_{aij}(dt)$$
    where $\widetilde{M}_{aij}(dt)=N_{aij}(dt)-Y_{ai}(t)e^{\beta_0^TZ_{aij}}A_{0ij}(dt)$ is defined on the calendar scale and $L(t)$ is the backwards recurrence time. Thus, we have
    \begin{align*}
        \mathbb{E}M_{aij}^4(u,x)&=\mathbb{E}\left(M_{aij}(x)-M_{aij}(u)\right)^4 \\
        &\le  C\mathbb{E}\left( M_{aij}^4(x)+M_{aij}^4(u) \right)\\
        &\le^{(*)} C \mathbb{E}\left( [\widetilde{M}_{aij}]_\tau^2\right)\\
        &=C\mathbb{E}\left( N_{aij}(\tau)\right)\\
        &=C \int_0^\tau s^{(0)}_{ij}(u,\beta_0)A_{0ij}(du)
    \end{align*}
    where $C$ is a universal constant, the second line is from $C_r$-inequality, the third line is due to Burkholder-Davis-Gundy's inequality (this needs a little bit of lengthy argument, see Lemma~\ref{lemma:bound}), and the last line follows from Condition (v). Hence, the first term is negligible when $m$ is large. The second term is bounded by
    \begin{align*}
    \mathbb{E}M_{aij}^2(u,x)\mathbb{E}M_{cij}^2(x,y)&=\int_u^x s^{(0)}_{ij}(u,\beta_0)A_{0ij}(du)\int_x^y s^{(0)}_{ij}(u,\beta_0)A_{0ij}(du)\\
    &=\sup_{v\in[0,\tau]}s^{(0)}_{ij}(v,\beta_0)\left(A_{0ij}(y)-A_{0ij}(u)\right)^2
    \end{align*}
    
    Finally,
    \begin{align*}
        \mathbb{E}\left(|M_{ij}(x)-M_{ij}(u)|^2|M_{ij}(y)-M_{ij}(x)|^2\right)&\le C\left(A_{0ij}(y)-A_{0ij}(u)\right)^2
    \end{align*}
    so that Billingsley's condition is satisfied with $r=1$ and $F(y)=\sqrt{C} A_{0ij}(y)$.
\end{proof}
\begin{lemma}[Bounding the fourth moment]\label{lemma:bound}  The fourth moment of $M_{ij}$ evaluated at $x$ is bounded by
$$M^4_{ij}(x)\le \widetilde{M}^4_{ij}(\tau).$$
As a consequence, the application of Burkholder-Davis-Gundy shows that
$$\mathbb{E}\widetilde{M}_{ij}^4(\tau)\le C\mathbb{E}[\widetilde{M}_{ij}]_\tau^2$$
where $C$ is a universal constant.
\end{lemma}
\begin{proof}
    Inspired by formula (26) in \citep{gill1980nonparametric}, we note that both $M^{4}_{ij}(dx)$ and $\widetilde{M}^4_{ij}(dt)$ are non-negative discrete measures with the same total number of jumps. So we have
\begin{align*}
    \int_0^\infty \mathbb{I}(u\le x)M_{ij}^4(du)&=\sum_{u\le x}M_{ij}^{4}(\Delta u)\\
    &\le\sum_{L(t)\le x}\widetilde{M}_{ij}^{4}(\Delta t)\\
    &=\int_0^\infty\mathbb{I}(L(t)\le x)\widetilde{M}_{ij}^4(dt)\\
    &\le\int_0^\infty \mathbb{I}(L(t)\le \tau)\widetilde{M}^4_{ij}(dt)\\
    &=\widetilde{M}^4_{ij}(\tau).
\end{align*}
The second inequality needs further explanation. First let us consider the second moment $M_{ij}^2(\Delta u)$ (drop the subscript $ij$) and denote $\Lambda(u)=N(u)-M(u)$:
\begin{align*}
   \sum_{u\le x}M^2(\Delta u)&=\sum_{u\le x}(N(u)-\Lambda(u))^2-(N(u-)-\Lambda(u-))^2\\
   &=\sum_{u\le x}N^2(\Delta u)-2N(\Delta u)\Lambda(u)\\
   &=\sum_{u\le x}N^2(\Delta u)-2\Lambda(u)\\
    &\le \sum_{L(t)\le x}\left[\widetilde{N}^2(\Delta t)-2\widetilde{\Lambda}(t)\right]\\
    &=\sum_{L(t)\le x} \widetilde{M}^2(\Delta t).
\end{align*}
The inequality in the fourth line is because for each $\Delta u$, there is a corresponding $\Delta t$ such that $N^2(\Delta u)$ is bounded by $\widetilde{N}^2(\Delta t)$. We use the following expansions to handle fourth moment:
\begin{align*}
    (a-b)^4&=a^4-4a^3b+6a^2b^2-4ab^3+b^4,\\
    a^2-(a-1)^2&=2a-1,\\
    a^3-(a-1)^3&=3a^2-3a+1,\\
    a^4-(a-1)^4&=4a^3-6a^2+4a-1.
\end{align*}
The first equality is for expanding $M^4(u)$, and the other three are used for expanding $N^4(\Delta u),N^3(\Delta u)$ and $N^2(\Delta u)$. After some algebra, we have

\begin{align*}
    {M}^4(\Delta u)&={N}^4(\Delta u)-4{N}^3(\Delta u)\Lambda(u)+6{N}^2(\Delta u)\Lambda^2(u)-4{N}(\Delta u)\Lambda^3(u).
\end{align*}
If ${N}(u-)=n-1$ and ${N}(u)=n$, then by some algebra, we have
$${M}^4(\Delta u)=4n^3-(6+12a)n^2+(4-12a+12a^2)n+4a-6a^2-4a^3$$
where $a=\Lambda(u)$. Now $\sum_{u\le x}M^4(\Delta u)\le \sum_{L(t)\le x}\widetilde{M}^4(\Delta t)$ is proved if we can show that
$$f(n,a)=4n^3-(6+12a)n^2+(4-12a+12a^2)n\ge 0$$
since if so, then each term $M^4(\Delta u)$ is bounded by a corresponding $\widetilde{M}^4(\Delta t)$. The above function is minimized at $a=\frac{n+1}{2}$ and plug-in we find
$$f(n,a)=4n^3-12n^2+n\ge 0$$
when $n\ge 12$.
\end{proof}

We are ready to show the weak convergence of $\widehat{U}$ and $\widehat{\Psi}$.
\begin{lemma}
    The processes
    $m^{-1/2}\widehat{U}(x,\beta_0)$ and $\widehat{\Psi}_{ij}(x,\beta_0)$ converge weakly to 
    $$\sum_{i\not=j}W^{(1)}_{ij}(x,\beta_0)-\sum_{i\not=j}\int_0^x e_{ij}(u,\beta_0)W^{(0)}_{ij}(du,\beta_0)\text{ and }\int_0^x\frac{W_{ij}^{(0)}(du,\beta_0)}{s_{ij}^{(0)}(u,\beta_0)},$$
    respectively. Further, components of the limiting distribution are independent of each other.
\end{lemma}
\begin{proof}
    This is by the functional version of continuous mapping theorem \citep{billingsley2013convergence}. The asymptotically independent components can be shown from the fact that components of $W_{ij}^{(k)}$ are uncorrelated.
\end{proof}

\subsubsection{Weak Convergence of $\hat{\beta}$ and $\hat{A}$}
\begin{proof}[Proof of Theorem~\ref{thm:weak_conv}]
    By algebra, we have
    \begin{align*}
        Cov(U(x,\beta_0),U(y,\beta_0))&=\mathbb{E}[U(x,\beta_0)U(y,\beta_0)]\\
        &=\sum_{i\not=j}\int_0^{x\vee y}s^{(2)}_{ij}(u,\beta_0)A_{0ij}(du)\\&-2\sum_{i\not=j}\int_0^{x\vee y}s^{(1)}_{ij}(u,\beta_0)e_{ij}^T(u,\beta_0)A_{0ij}(du)\\
        &+\sum_{i\not=j}\int_0^xe_{ij}(u,\beta_0)s^{(0)}_{ij}(u,\beta_0)e_{ij}^T(u,\beta_0)A_{0ij}(du)\\
        &=\sum_{i\not=j}\int_0^{x\vee y}\left(s^{(2)}_{ij}(u,\beta_0)-\frac{s^{(1)}_{ij}(u,\beta_0)^{\otimes 2}}{s^{(0)}_{ij}(u,\beta_0)}\right)A_{0ij}(du)\\
        &=\Sigma(x,\beta_0).
    \end{align*}
    The convergence of $\mathscr{I}(x,\beta^*)$ to $\Sigma(x,\beta)$ follows from Condition (ii). Hence, by condition (iii), we have
    $$\sqrt{m}\left(\widehat{\beta}-\beta_0\right)\Rightarrow \mathcal{N}\left(0,\Sigma(\tau,\beta_0)^{-1}\right).$$
    Next, consistency of $\widehat{A}_{0ij}$ is simple and follows directly from functional LLN. Hence,
    $$\widehat{\eta}_{ij}(x,\beta^*)\rightarrow -\int_0^x e_{ij}(u,\beta_0)A_{0ij}(du):=\eta_{ij}(x,\beta_0)$$
    uniformly in probability. By Lemma~\ref{lemma:representation}, we have
    $$\sqrt{m}\left(\widehat{A}_{0ij}(x,\widehat{\beta})-A_{0ij}(x,\beta_0))\right)\Rightarrow\Psi_{ij}(x,\beta_0)+\eta_{ij}(x,\beta_0)^T\mathcal{N}\left(0,\Sigma(\tau,\beta_0)^{-1}\right).$$
    
\end{proof}

\section{Acknowledgments}

We sincerely thank Dr. Dabrowska for her invaluable comments on semi-Markov and Markov renewal models and their implementations through personal communication.

\bibliographystyle{imsart-nameyear}  

\bibliography{references, pkg, historical}

\end{document}